\documentclass[a4paper,USenglish,cleveref, autoref, thm-restate]{lipics-v2021}

\title{Synthesizing Dominant Strategies for Liveness (Full Version)}

\author{Bernd Finkbeiner}{CISPA Helmholtz Center for Information Security, Saarbrücken, Germany}{finkbeiner@cispa.de}{https://orcid.org/0000-0002-4280-8441}{}
\author{Noemi Passing}{CISPA Helmholtz Center for Information Security, Saarbrücken, Germany}{noemi.passing@cispa.de}{https://orcid.org/0000-0001-7781-043X}{}
\authorrunning{B. Finkbeiner and N. Passing}

\Copyright{Bernd Finkbeiner and Noemi Passing}


\category{} 

\relatedversion{A conference version of this paper is available at~\cite{FinalVersion}.}

\funding{This work was supported by~DFG grant 389792660 as part of~TRR~248~(CPEC) and by~ERC grant 683300 (OSARES). N.~Passing is a PhD candidate at Saarland University, Germany.}

\nolinenumbers

\hideLIPIcs

\usepackage[utf8]{inputenc}
\usepackage{csquotes}
\usepackage{ltl}
\usepackage[color=teal!50]{todonotes}
\usepackage{amssymb}
\usepackage{tikz}
\usetikzlibrary{arrows,arrows.meta,automata,shapes,shapes.multipart,positioning,calc,fit,backgrounds}
\usepackage[ruled,vlined]{algorithm2e}
\usepackage{mathtools}
\usepackage{subcaption}
\usepackage{proof}

\usepackage[unicode]{hyperref}

\newcommand{\Next}{\LTLnext}
\newcommand{\Globally}{\LTLsquare}
\newcommand{\Eventually}{\LTLdiamond}
\newcommand{\Until}{\LTLuntil}

\newcommand{\true}{\mathit{true}}
\newcommand{\false}{\mathit{false}}

\newcommand{\ie}{i.e.\@\xspace}
\newcommand{\eg}{e.g.\@\xspace}

\newcommand{\Lang}[1]{\mathcal{L}(#1)}

\newcommand{\unprimeName}{\mathit{unpr}}
\newcommand{\unprime}[1]{\unprimeName(#1)}
\newcommand{\primeSequenceName}{\mathit{pr}}
\newcommand{\primeSequence}[1]{\primeSequenceName(#1)}

\newcommand{\arena}{\mathbb{A}}
\newcommand{\win}{W}

\newcommand{\compatible}[2]{\operatorname{cpbl}(#1,#2)}
\newcommand{\projp}[1]{f_\mathit{alt}(#1)}
\newcommand{\projq}[1]{f_\mathit{dom}(#1)}

\newcommand{\dd}{\mathit{dd}}

\newcommand{\iDDominance}{Delay-Dominance\@\xspace}
\newcommand{\iDdominance}{Delay-dominance\@\xspace}
\newcommand{\iDDominant}{Delay-Dominant\@\xspace}

\newcommand{\iddominance}{delay-dominance\@\xspace}
\newcommand{\iddominant}{delay-dominant\@\xspace}
\newcommand{\iddominates}{delay-dominates\@\xspace}
\newcommand{\iddominate}{delay-dominate\@\xspace}

\newcommand{\dom}{Duplicator\@\xspace}
\newcommand{\alt}{Spoiler\@\xspace}

\newcommand{\dominates}[2]{#2 \preceq #1}

\newcommand{\ddominatesNotation}[3]{#2 \trianglelefteq^{#3} #1}
\newcommand{\ddominatesSequence}[4]{#2 \trianglelefteq^{#4}_{#3} #1}
\newcommand{\ddominatesSequenceNegated}[4]{#2 \not\trianglelefteq^{#4}_{#3} #1}

\newcommand{\ddACA}[1]{\mathcal{B}^\mathit{A}_{\mathcal{A}_#1}}
\newcommand{\ddACAComponent}[1]{{#1}^\mathit{A}}
\newcommand{\ddUCAnonProj}[1]{\mathcal{B}^\mathit{U}_{\mathcal{A}_#1}}

\newcommand{\ddUCA}[1]{\mathcal{A}^\dd_{\mathcal{A}_#1}}
\newcommand{\aca}[1]{\mathcal{A}}

\newcommand{\pc}{\mathbin{||}}

\newcommand{\coloneq}{:=}

\newcommand{\pref}[2]{{#1_{|#2}}}
\newcommand{\projection}[2]{#2(#1)}

\newcommand{\computation}[2]{\mathit{comp}(#1,#2)}

\newcommand{\children}[1]{\operatorname{c}(#1)}

\newcommand{\successor}[3]{\vartheta(#1,#2,#3)}

\newcommand{\allVariables}{\Sigma}
\newcommand{\allInputs}{I}
\newcommand{\allOutputs}{O}

\newcommand{\inputs}[1]{{I_{#1}}}
\newcommand{\outputs}[1]{{O_{#1}}}
\newcommand{\primedOutputs}[1]{{O'_{#1}}}
\newcommand{\primedVariables}[1]{{\Sigma'_{#1}}}
\newcommand{\variables}[1]{{\Sigma_{#1}}}
\newcommand{\inpFunc}{{\mathit{inp}}}
\newcommand{\outFunc}{{\mathit{out}}}

\newcommand{\env}{\mathit{env}}
\newcommand{\sysProc}{P^-\!}

\DeclarePairedDelimiter\myVec{\langle}{\rangle}

\tikzset{
    pics/vhsplit/.style n args = {4}{
        code = {
        \node[inner sep=4pt] (B) at (0,0) {#2\vphantom{#3}};
        \node[inner sep=4pt] (C) [right = 0cm of B] {#3\vphantom{#2}};
        \node[anchor=south,inner sep=2.5pt] (A) at ($(B.north west)!0.5!(C.north east)$) {#1};
        \begin{scope}[on background layer]
        	\node[inner sep=2pt,draw,rounded corners=5pt,fit=(A)(B)(C),fill=#4] {}; 
       		\draw ([xshift=-2pt]B.north west) -- ([xshift=2pt]C.north east)
              ([yshift=-2pt]B.south east) -- (C.north west)
              (B.north east) -- (C.south west);    
        \end{scope}
        }
    }
}

\begin{document}

\maketitle

\begin{abstract}
Reactive synthesis automatically derives a strategy that satisfies a given specification. However, requiring a strategy to meet the specification in every situation is, in many cases, too hard of a requirement. Particularly in compositional synthesis of distributed systems, individual winning strategies for the processes often do not exist. Remorsefree dominance, a weaker notion than winning, accounts for such situations: dominant strategies are only required to be as good as any alternative strategy, \ie, they are allowed to violate the specification if no other strategy would have satisfied it in the same situation.
The composition of dominant strategies is only guaranteed to be dominant for safety properties, though; preventing the use of dominance in compositional synthesis for liveness specifications. Yet, safety properties are often not expressive enough. 
In this paper, we thus introduce a new winning condition for strategies, called \iddominance, that overcomes this weakness of remorsefree~dominance: we show that it is compositional for many safety and liveness specifications, enabling a compositional synthesis algorithm based on \iddominance for general specifications. Furthermore, we introduce an automaton construction for recognizing \iddominant strategies and prove its soundness and completeness.
The resulting automaton is of single-exponential size in the squared length of the specification and can immediately be used for safraless synthesis procedures. Thus, synthesis of \iddominant strategies is, as synthesis of winning strategies, in 2EXPTIME.
\end{abstract}

\section{Introduction}\label{sec:introduction}

Reactive synthesis is the task of automatically deriving a strategy that satisfies a formal specification, \eg, given in LTL~\cite{Pnueli77}, in \emph{every} situation. Such strategies are called winning.
In many cases, however, requiring the strategy to satisfy the specification in every situation is too hard of a requirement. A prominent example is the compositional synthesis of~distributed systems consisting of several processes.
Compositional approaches for distributed synthesis~\cite{KupfermanPV06,FiliotJR10,FinkbeinerGP21,FinkbeinerGP22,FinkbeinerP22} break down the synthesis task for the whole system into several smaller ones for the individual processes.
This is necessary due to the general undecidability~\cite{PnueliR90} of distributed synthesis and the non-elementary complexity~\cite{FinkbeinerS05} for decidable cases: non-compositional distributed synthesis approaches~\cite{FinkbeinerS13,FinkbeinerS07SMT} suffer from a severe state space explosion problem and are thus not feasible for larger systems.
However, winning strategies rarely exist when considering the processes individually in the smaller subtasks of compositional synthesis since usually the processes need to collaborate in order to achieve the overall system's correctness. For instance, a particular input sequence may prevent the satisfaction of the specification no matter how a single process reacts, yet, the other processes of the system ensure in the interplay of the whole system that this input sequence will never be produced.

\emph{Remorsefree dominance}~\cite{DammF11}, a weaker notion than winning, accounts for such situations. A dominant strategy is allowed to violate the specification as long as no other strategy would have satisfied it in the same situation. Hence, a dominant strategy is a best-effort strategy as we do not blame it for violating the specification if the violation is not its fault.
Searching for dominant strategies rather than winning ones allows us to find strategies that do not necessarily satisfy the specification in all situations but in all that are \emph{realistic} in the sense that they occur in the interplay of the processes if all of them play best-effort strategies.

The parallel composition of dominant strategies, however, is only guaranteed to be~dominant for safety properties~\cite{DammF14}. For liveness specifications, in contrast, dominance is not~a compositional notion and thus not suitable for compositional synthesis. Consider, for example, a system with two processes $p_1$ and $p_2$ sending messages to each other, denoted by atomic propositions $m_1$ and $m_2$, respectively. Both processes are required to send their message eventually, \ie, $\varphi = \Eventually m_1 \land \Eventually m_2$. For~$p_i$, it is dominant to wait for the other process to send the message $m_{3-i}$ before sending its own message $m_i$: if $p_{3-i}$ sends its message eventually, $p_i$ does so as well, satisfying $\varphi$. If $p_{3-i}$ never sends its message, $\varphi$ is violated, no matter how~$p_i$ reacts, and thus the violation of $\varphi$ is not $p_i$'s fault. Combining these strategies for $p_1$ and $p_2$, however, yields a system that never sends any message since both processes wait indefinitely for each other, while there clearly exist strategies for the whole system that satisfy $\varphi$.

\emph{Bounded dominance}~\cite{DammF14} is a variant of remorsefree dominance that ensures compositionality of general properties. Intuitively, it reduces every specification $\varphi$ to a safety property by introducing a measure of the strategy's progress with respect to $\varphi$, and by bounding the number of non-progress steps, \ie, steps in which no progress is made.
Yet, bounded dominance has two major disadvantages: (i) it requires a concrete bound on the number of non-progress steps, and~(ii)~not every bounded dominant strategy is dominant: if the~bound~$n$ is chosen too small, every strategy, also a non-dominant one, is trivially $n$-dominant.

In this paper, we introduce a new winning condition for strategies, called \emph{\iddominance}, that builds upon the ideas of bounded dominance but circumvents the aforementioned weaknesses. Similar to bounded dominance, it introduces a progress measure on strategies. However, it does not require a concrete bound on the number of non-progress steps but relates such steps in the potentially \iddominant strategy~$s$ to non-progress steps in an alternative strategy~$t$: intuitively, $s$ \iddominates $t$ if, whenever $s$ makes a non-progress step, $t$ makes a non-progress step \emph{eventually} as well.
A strategy $s$ is then \iddominant if it \iddominates every other strategy $t$.
In this way, we ensure that a \iddominant strategy satisfies the specification \enquote{faster} than all other strategies in all situations in which the specification can be satisfied.
\iDdominance considers specifications given as alternating co-Büchi automata. Non-progress steps with respect to the automaton are those that enforce a visit of a rejecting state in all run trees. We introduce a two-player game, the so-called \emph{\iddominance game}, 
which is vaguely leaned on the delayed simulation game for alternating Büchi automata~\cite{FritzW05}, 
to formally define \iddominance: the winner of the game determines whether or not a strategy $s$ \iddominates a strategy $t$ on a given input sequence.

\enlargethispage{1.5\baselineskip}
We~(i)~show that every \iddominant strategy is also remorsefree dominant, and~(ii)~introduce a criterion for automata such that, if the criterion is satisfied, compositionality of \iddominance is guaranteed. The criterion is satisfied for many automata; both ones describing safety properties and ones describing liveness properties. Thus, \iddominance overcomes the weaknesses of both remorsefree and bounded dominance.
Note that since \iddominance relies, as bounded dominance, on the automaton structure, there are realizable specifications for which no \iddominant strategy exists. Yet, we experienced that this rarely occurs in practice when constructing the automaton from an LTL formula with standard algorithms.
Moreover, if a \iddominant strategy exists, it is guaranteed to be winning if the specification is realizable.
Hence, the parallel composition of \iddominant strategies for all processes in a distributed system is winning for the whole system as long as the specification is realizable and as long as the compositionality criterion is satisfied.
Therefore, \iddominance is a suitable notion for compositional synthesis.

We thus introduce a synthesis approach for \iddominant strategies that immediately enables a compositional synthesis algorithm for distributed systems, namely synthesizing \iddominant strategies for the processes separately.
We present the construction of a universal co-Büchi automaton $\ddUCA{\varphi}$ from an LTL formula $\varphi$ that recognizes \iddominant strategies.
$\ddUCA{\varphi}$ can immediately be used for safraless synthesis~\cite{KupfermanV05} approaches such as bounded synthesis~\cite{FinkbeinerS13} to synthesize \iddominant strategies.
We show that the size of $\ddUCA{\varphi}$ is single-exponential in the squared length of~$\varphi$. Thus, synthesis of \iddominant strategies is, similar to synthesis of winning or remorsefree dominant strategies, in 2EXPTIME.

\subparagraph*{Related Work.}

\emph{Remorsefree dominance} has first been introduced for reactive synthesis in~\cite{DammF11}. Dominant strategies have been utilized for compositional synthesis of safety properties~\cite{DammF14}. Building up on this work, a compositional synthesis algorithm, that finds solutions in more cases by incrementally synthesizing individual dominant strategies, has been developed~\cite{FinkbeinerP20}.
Both algorithms suffer from the non-compositionality of dominant strategies for liveness properties. 
\emph{Bounded dominance}~\cite{DammF14}, a variant of dominance that introduces a bound on the number of steps in which a strategy does not make progress with respect to the specification, solves this problem. However, it requires a concrete bound on the number of non-progress steps. Moreover, a bounded dominant strategy is not necessarily dominant.

\emph{Good-enough synthesis}~\cite{AlmagorK20,LiTVZ21} follows a similar idea as dominance. It is thus not compositional for liveness properties either. In good-enough synthesis, conjuncts of the specification can be marked as \emph{strong}. If the specification is unrealizable, a good-enough strategy needs to satisfy the strong conjuncts while it may violate the other ones.
Thus, dominance can be seen as the special case of good-enough synthesis in which no conjuncts are marked as strong.
Good-enough synthesis can be extended to a multi-valued correctness notion~\cite{AlmagorK20}.

\emph{Synthesis under environment assumptions} is a well-studied problem that also aims at relaxing the requirements on a strategy. There, explicit assumptions on the environment are added to the specification.
These assumptions can be LTL formulas restricting the possible input sequences (see, \eg, \cite{ChatterjeeHJ08,BloemEJK14}) or environment strategies~(see, \eg, \cite{AminofGMR18,AminofGR21,FinkbeinerP21,FinkbeinerP22}). The assumptions can also be conceptual such as assuming that the environment is rational (see, \eg, \cite{FismanKL10,KupfermanPV14,BrenguierRS15, ConduracheFGR16}). Synthesis under environment assumptions is orthogonal to the synthesis of dominant strategies and good-enough synthesis since it requires an explicit assumption on the environment, while the latter two approaches rely on implicit assumptions.


\section{Preliminaries}\label{sec:preliminaries}

\subparagraph*{Notation.}
Given an infinite word $\sigma = \sigma_0 \sigma_1 \ldots \in (2^\allVariables)^\omega$, we denote the prefix of length $t+1$ of~$\sigma$ with $\pref{\sigma}{t} := \sigma_{0} \ldots \sigma_{t}$. For $\sigma$ and a set $X\subseteq\allVariables$, let $\sigma \cap X := (\sigma_0 \cap X)(\sigma_1 \cap X)\ldots\in (2^X)^\omega\!$.
For~$\sigma \in (2^\allVariables)^\omega$, $\sigma' \in (2^{\allVariables'})^\omega$ with $\allVariables \cap \allVariables' = \emptyset$, we define $\sigma \cup \sigma' := (\sigma_0 \cup \sigma'_0) (\sigma_1 \cup \sigma'_1) \ldots \in (2^{\Sigma \cup \Sigma'})^\omega\!$.
For a $k$-tuple $a$, we denote the $j$-th component of $a$ with $\projection{a}{j}$.
We represent a Boolean~formula $\bigvee_{i} \bigwedge_{j} c_{i,j}$ in disjunctive normal form~(DNF) also in its set notation $\bigcup_{i} \{ \bigcup_{j} \{c_{i,j}\} \}$.

\enlargethispage{1.5\baselineskip}
\subparagraph*{LTL.}
Linear-time temporal logic~(LTL)~\cite{Pnueli77} is a standard specification language for linear-time properties. 
Let $\Sigma$ be a finite set of atomic propositions and let $a \in \Sigma$. The syntax of LTL is given by
$ \varphi, \psi ::= a ~ | ~ \neg \varphi ~ | ~ \varphi \lor \psi ~ | ~ \varphi \land \psi ~ | ~ \Next \varphi ~ | ~ \varphi \Until \psi$.
We define $\true = a \lor \neg a$, $\false = \neg \true$, $\Eventually \varphi = \true \Until \varphi$, and $\Globally \varphi = \neg \Eventually \neg \varphi$ as usual. We use the standard semantics.
The language~$\Lang{\varphi}$ of an LTL formula $\varphi$ is the set of infinite words that satisfy $\varphi$.

\subparagraph*{Non-Alternating \texorpdfstring{\boldmath $\omega$}{omega}-Automata.}
Given a finite alphabet~$\Sigma$, a Büchi (resp.\ co-Büchi) automaton $\mathcal{A} = (Q,Q_0,\delta,F)$ over $\Sigma$ consists of a finite set of states $Q$, an initial state $q_0 \in Q$, a transition relation $\delta: Q \times 2^\Sigma \times Q$, and a set of accepting (resp.\ rejecting) states $F \subseteq Q$.
For an infinite word $\sigma = \sigma_0\sigma_1 \ldots \in (2^\Sigma)^\omega\!$, a run of $\mathcal{A}$ induced by $\sigma$ is an infinite sequence $q_0 q_1 \ldots \in Q^\omega\!$ of states with $(q_i,\sigma_i,q_{i+1}) \in \delta$ for all $i \geq 0$. 
A run is accepting if it contains infinitely many accepting states (resp.\ only finitely many rejecting states).
A nondeterministic (resp.\ universal) automaton~$\mathcal{A}$ accepts a word $\sigma$ if some run is accepting (resp.\ all runs are accepting).
The language $\Lang{\mathcal{A}}$ of $\mathcal{A}$ is the set of all accepted words. We consider nondeterministic Büchi automata (NBAs) and universal co-Büchi automata (UCAs).

\subparagraph*{Alternating \texorpdfstring{\boldmath $\omega$}{omega}-Automata.}
An alternating Büchi (resp.\ co-Büchi) automaton (ABA resp.\ ACA) $\mathcal{A} = (Q,q_0,\delta,F)$ over a finite alphabet $\Sigma$ consists of a finite set of states $Q$, an initial state $q_0 \subseteq Q$, a transition function $\delta: Q \times 2^\Sigma \rightarrow \mathbb{B}^+(Q)$, where $\mathbb{B}^+(Q)$ is the set of positive Boolean formulas over~$Q$, and a set of accepting (resp.\ rejecting) states $F \subseteq Q$.
We assume that the elements of $\mathbb{B}^+(Q)$ are given in DNF.
Runs of $\mathcal{A}$ are $Q$-labeled trees: a tree $\mathcal{T}$ is a prefix-closed subset of~$\mathbb{N}^*\!$. 
The children of a node $x\in\mathcal{T}$ are $\children{x} = \{ x \cdot d \in \mathcal{T} \mid d \in \mathbb{N} \}$. 
An $X$-labeled tree $(\mathcal{T},\ell)$ consists of a tree $\mathcal{T}$ and a labeling function $\ell: \mathcal{T} \rightarrow X$. 
A branch of $(\mathcal{T},\ell)$ is a maximal sequence $\ell(x_0) \ell(x_1) \ldots$ with~$x_0 = \varepsilon$ and $x_{i+1} \in \children{x_i}$ for $i \geq 0$.
A run tree of $\mathcal{A}$ induced by $\sigma \in (2^\Sigma)^\omega$ is a $Q$-labeled tree $(\mathcal{T},\ell)$ with $\ell(\varepsilon) = q_0$ and, for all $x \in \mathcal{T}$, $\{\ell(x') \mid x' \in \children{x}\} \in \delta(\ell(x),\sigma_{|x|})$.
A run tree is accepting if every infinite branch contains infinitely many accepting states (resp.\ only finitely many rejecting states).~$\mathcal{A}$ accepts~$\sigma$ if there is some accepting run tree.
The language $\Lang{\mathcal{A}}$ of $\mathcal{A}$ is the set of all accepted words.

\subparagraph*{Two-Player Games.} An arena is a tuple $\arena = (V, V_0, V_1, v_0, E)$, where $V$, $V_0$, $V_1$ are finite sets of positions with $V = V_0 \cup V_1$ and $V_0 \cap V_1 = \emptyset$, $v_0 \in V$ is the initial position, $E \subseteq V \times V$ is a set of edges such that $\forall v \in V.~ \exists v' \in V.~ (v,v') \in E$. Player~$i$ controls positions in $V_i$.
A game $\mathcal{G} = (\arena, \win)$ consists of an arena $\arena$ and a winning condition $\win \subseteq V^\omega\!$. 
A play is an infinite sequence $\rho \in V^\omega\!$ such that $(\rho_i, \rho_{i+1}) \in E$ for all $i \in \mathbb{N}$. The player owning a position chooses the edge on which the play is continued. 
A~play~$\rho$ is initial if $\rho_0 = v_0$ holds.
It is winning for Player 0 if $\rho \in \win$ and winning for Player 1 otherwise.
A strategy for Player~$i$ is a function $\tau: V^*V_i \rightarrow V$ such that $(v,v') \in E$ whenever $\tau(w,v) = v'$ for some $w \in V^*\!$, $v \in V_i$. 
A play~$\rho$ is consistent with a strategy $\tau$ if, for all $j \in \mathbb{N}$, $\rho_j \in V_i$ implies $\rho_{j+1} = \tau(\pref{\rho}{j})$. 
A strategy for Player $i$ is winning if all initial and consistent plays are winning for Player~$i$.

\subparagraph*{System Architectures.}
An architecture is a tuple $A=(P, \allVariables, \inpFunc, \outFunc)$, where $P$ is a set of processes consisting of the environment~$\env$ and a set $\sysProc = P \setminus\{env\}$ of $n$ system processes,~$\allVariables$ is a set of Boolean variables, $\inpFunc = \myVec{I_1, \dots, I_n}$ assigns a set $\inputs{j} \subseteq \allVariables$ of input variables to each $p_j \in \sysProc$, and $\outFunc = \myVec{O_\env, O_1, \dots O_n}$ assigns a set $\outputs{j} \subseteq \Sigma$ of output variables to each $p_j \in P$.
For all $p_j, p_k \in \sysProc$ with $j \neq k$, $\inputs{j} \cap \outputs{j} = \emptyset$ and $\outputs{j} \cap \outputs{k} = \emptyset$ hold.
The variables~$\variables{j}$ of $p_j \in \sysProc$ are given by $\variables{j} = \inputs{j} \cup \outputs{j}$.
The inputs~$\allInputs$, outputs~$\allOutputs$, and variables~$\allVariables$ of the whole system are defined by $X = \bigcup_{p_j \in \sysProc} X_j$ for $X \in \{I, O, \allVariables\}$.
$A$ is called distributed if $|\sysProc| \geq 2$.
In the remainder of this paper, we assume that a distributed architecture is given.

\subparagraph*{Process Strategies.}
A strategy for process $p_i$ is a function $s_i: (2^\inputs{i})^* \rightarrow 2^\outputs{i}$ mapping a history of inputs to outputs. 
We model~$s_i$ as a Moore machine~$\mathcal{M}_i = (T, t_0, \tau, o)$ consisting of a finite set of states $T$, an initial state $t_0 \in T$, a transition function $\tau: T \times 2^\inputs{i} \rightarrow T$, and a labeling function $o: T \rightarrow 2^\outputs{i}$.
For a sequence $\gamma = \gamma_0 \gamma_1 \dotsc \in (2^{\inputs{i}})^\omega\!$, $\mathcal{M}_i$ produces a path $(t_0 , \gamma_0 \cup o(t_0)) (t_1 , \gamma_1 \cup o(t_1)) \dotsc \in (T \times 2^{\inputs{i}\cup\outputs{i}})^\omega\!$, where $\tau(t_j, \gamma_j) = t_{j+1}$. 
The projection of a path to the variables is called a trace.
The trace produced by $\mathcal{M}_i$ on $\gamma \in (2^\inputs{i})^\omega$ is called the computation of~$s_i$ on $\gamma$, denoted $\computation{s_i}{\gamma}$. 
We say that $s_i$ is winning for an LTL formula~$\varphi$, denoted $s_i \models \varphi$, if $\computation{s_i}{\gamma} \models \varphi$ holds for all input sequences $\gamma \in (2^{\inputs{i}})^\omega\!$.
Overloading notation with two-player games, we call a process strategy simply a strategy whenever the context is clear.
The parallel composition $\mathcal{M}_i \pc \mathcal{M}_j$ of two Moore machines $\mathcal{M}_i = (T_i,t^i_0,\tau_i,o_i)$, $\mathcal{M}_j = (T_j,t^j_0,\tau_j,o_j)$ for $p_i, p_j\in\sysProc$ is the Moore machine $(T,t_0,\tau,o)$ with inputs $(I_i \cup I_j)\setminus(O_i\cup O_j)$ and outputs $O_i \cup O_j$ as well as $T = T_i \times T_j$, $t_0=(t^i_0,t^j_0)$, 
$\tau((t,t'),\iota) = (\tau_i(t,(\iota \cup o_j(t')) \cap \inputs{i}),\tau_j(t',(\iota \cup o_i(t))\cap\inputs{j}))$, and 
$o((t,t')) = o_i(t) \cup o_j(t')$.

\subparagraph*{Synthesis.}
Given a specification $\varphi$, synthesis derives strategies $s_1, \dots, s_n$ for the system processes such that $s_1 \pc \dots \pc s_n \models \varphi$, \ie, such that the parallel composition of the strategies satisfies~$\varphi$ for all input sequences generated by the environment. 
If such strategies exist,~$\varphi$ is called realizable. Bounded synthesis~\cite{FinkbeinerS13} additionally bounds the size of the strategies. 
The search for strategies is encoded into a constraint system that is satisfiable if, and only if, $\varphi$ is realizable for the size bound.
There are~SMT,~SAT, QBF, and DQBF encodings~\cite{FinkbeinerS13,FaymonvilleFRT17,Baumeister17}.
We consider a compositional synthesis approach that synthesizes strategies for the processes separately. Thus, outputs produced by the other system processes are treated similar to the environment outputs, namely as part of the input sequence of the considered process. Nevertheless, compositional synthesis derives strategies such that $s_1 \pc \dots \pc s_n \models \varphi$ holds. {}


\section{Dominant Strategies and Liveness Properties}\label{sec:dominance}

Given a  specification $\varphi$, the na\"ive compositional synthesis approach is to synthesize strategies $s_1, \dots, s_n$ for the system processes such that $s_i \models \varphi$ holds for all $p_i \in \sysProc$. Then, it follows immediately that $s_1 \pc \dots \pc s_n \models \varphi$ holds as well. However, since winning strategies are required to satisfy $\varphi$ for every input sequence, usually no such individual winning strategies exist due to complex interconnections in the system. Therefore, the na\"ive approach fails in many cases.
The notion of \emph{remorsefree dominance}~\cite{DammF11}, in contrast, has been successfully used in compositional synthesis~\cite{DammF14,FinkbeinerP20}. The main idea is to synthesize \emph{dominant} strategies for the system processes separately instead of winning ones.
Dominant strategies are, in contrast to winning strategies, allowed to violate the specification for some input sequence if no other strategy would have satisfied it in the same situation. Thus, remorsefree dominance is a weaker requirement than winning and therefore individual dominant strategies exist for more systems. Formally, remorsefree dominant strategies are defined as follows:

\begin{definition}[Dominant Strategy \cite{DammF14}]
Let $\varphi$ be an LTL formula. Let $s$ and $t$ be strategies for process~$p_i$. Then, $t$ \emph{is dominated by} $s$, denoted $\dominates{s}{t}$, if for all input sequences $\gamma \in (2^\inputs{i})^\omega\!$ either $\computation{s}{\gamma} \models \varphi$ or $\computation{t}{\gamma}\not\models\varphi$ holds.
Strategy $s$ is called \emph{dominant} for $\varphi$ if $\dominates{s}{t}$ holds for all strategies $t$ for process $p_i$.
\end{definition}

\enlargethispage{-1\baselineskip}
Intuitively, a strategy $s$ dominates a strategy $t$ if it is \emph{at least as good} as~$t$. It is dominant for $\varphi$ if it is at least as good as \emph{every other possible strategy} and thus if it is \emph{as good as possible}.
As an example, reconsider the message sending system. Let $s_i$ be a strategy for process $p_i$ that outputs~$m_i$ in the very first step. It satisfies $\varphi = \Eventually m_1 \land \Eventually m_2$ on all input sequences containing at least one~$m_{3-i}$. On all other input sequences, it violates~$\varphi$. Let~$t_i$ be some alternative strategy.
Since no strategy for~$p_i$ can influence~$m_{3-i}$,~$t_i$ satisfies $\varphi$ only on input sequences containing at least one $m_{3-i}$. Yet, $s_i$ satisfies~$\varphi$ for such sequences as well. Hence, $s_i$ dominates $t_i$ and since we chose $t_i$ arbitrarily, $s_i$ is dominant for $\varphi$.

Synthesizing dominant strategies rather than winning ones allows us to synthesize strategies for the processes of a distributed system compositionally, although no winning strategies for the individual processes exist.
Dominant strategies for the individual processes can then be recomposed to obtain a strategy for the whole system. 
For safety specifications, the composed strategy is guaranteed to be dominant for the specification as well:

\begin{theorem}[Compositionality of Dominance for Safety Properties~\cite{DammF14}]\label{thm:compositionality_safety}
Let $\varphi$ be an LTL formula. Let $s_1$ and $s_2$ be dominant strategies for processes $p_1$ and $p_2$, respectively, as well as for $\varphi$. If $\varphi$ is a safety property, then $s_1 \pc s_2$ is dominant for $p_1 \pc p_2$ and $\varphi$.
\end{theorem}

Compositionality is a crucial property for compositional synthesis: it allows for concluding that the parallel composition of the separately synthesized process strategies is indeed a useful strategy for the whole system. Thus, \Cref{thm:compositionality_safety} enables compositional synthesis~with dominant strategies for safety properties.
For liveness properties, however, the parallel~composition of two dominant strategies is not necessarily dominant: consider strategy $t_i$ for~$p_i$ in the message sending system that waits for $m_{3-i}$ before sending its own message. This strategy is dominant for $\varphi$: for input sequences in which $m_{3-i}$ occurs eventually,~$t_i$ sends $m_i$ in the next step, satisfying $\varphi$. For all other input sequences, no strategy for~$p_i$ can satisfy~$\varphi$. Yet, the parallel composition of $t_1$ and $t_2$ does not send any message; violating~$\varphi$, while there exist strategies that satisfy $\varphi$, \eg, a strategy sending both $m_1$ and $m_2$ in the first step.

\emph{Bounded dominance}~\cite{DammF14} is a variant of dominance that is compositional for both safety and liveness properties. Intuitively, it reduces the specification $\varphi$ to a safety property by introducing a bound on the number of steps in which the strategy \emph{does not make progress} with respect to $\varphi$. The progress measure is defined on an equivalent UCA $\mathcal{A}$ for $\varphi$. The measure $m_\mathcal{A}$ of a process strategy $s$ on an input sequence $\gamma \in (2^\inputs{i})^\omega$ is then the supremum of the number of rejecting states of the runs of $\mathcal{A}$ induced by $\computation{s}{\gamma}$. Thus, a strategy~$s$ $n$-dominates a strategy $t$ for $\mathcal{A}$ and $n\in\mathbb{N}$ if for every $\gamma \in (2^\inputs{i})^\omega$, either $m_\mathcal{A}(\computation{s}{\gamma}) \leq n$ or $m_\mathcal{A}(\computation{t}{\gamma}) > n$ holds.
If~$\mathcal{A}$ is a safety automaton, then remorsefree dominance and bounded dominance coincide.
For liveness specifications, however, they differ.

Yet, bounded dominance does not imply dominance: there are specifications~$\varphi$ with a minimal measure $m$, \ie, all strategies have a measure of at least $m$~\cite{DammF14}.
When choosing a bound $n<m$, every strategy is trivially $n$-dominant for~$\varphi$, even non-dominant ones.
Hence, the choice of the bound is crucial for bounded dominance. It is not obvious how to determine a \emph{good} bound, though:
it needs to be large enough to avoid non-dominant strategies.
As the bound has a huge impact on the synthesis time, however, it cannot be chosen too large as otherwise synthesis becomes infeasible.
Especially for specifications with several complex dependencies between processes, it is hard to determine a proper bound.
Therefore, bounded dominance is not a suitable notion for compositional synthesis for liveness properties.
In the remainder of this paper, we introduce a different variant of dominance that implies remorsefree dominance and that ensures compositionality also for many liveness properties. {}


\section{\texorpdfstring{\iDDominance}{Delay-Dominance}}\label{sec:ddominance}

In this section, we introduce a new winning condition for strategies, \emph{\iddominance}, which resembles remorsefree dominance but ensures compositionality also for many liveness properties.
It builds on the idea of bounded dominance to not only consider the satisfaction of the LTL formula~$\varphi$ but to measure progress based on an automaton representation of~$\varphi$. Similar to bounded dominance, we utilize visits of rejecting states in a co-Büchi automaton. 
Yet, we use an \emph{alternating} automaton instead of a universal one. Note that \iddominance can be equivalently formulated on UCAs, yet, using ACAs allows for more efficient synthesis of \iddominant strategies (see \Cref{sec:automaton_construction}).
Moreover, we do not require a fixed bound on the number of visits to rejecting states; rather, we relate visits of rejecting states induced by the \iddominant strategy to visits of rejecting states induced by the alternative strategy. {}

Intuitively, \iddominance requires that every visit to a rejecting state in the ACA~$\mathcal{A}$ caused by the \iddominant strategy \emph{is matched} by a visit to a rejecting state caused by the alternative strategy \emph{eventually}.
The rejecting states of the ACA $\mathcal{A}$ are closely related to the satisfaction of the LTL specification $\varphi$: if infinitely many rejecting states are visited, then $\varphi$ is not satisfied. Thus, \iddominance allows a strategy to violate the specification if all alternative strategies violate it as well.
Defining \iddominance on the rejecting states of $\mathcal{A}$ instead of the satisfaction of $\varphi$ allows for measuring the progress on satisfying the specification. Thus, we can distinguish strategies that wait indefinitely for another process from those that do not: intuitively, a strategy $s$ that waits will visit a rejecting state later than a strategy $t$ that does not. This visit to a rejecting state is then not matched eventually by a visit to a rejecting state in $t$, preventing \iddominance of $s$.

Formally, we present a game-based definition for \iddominance: we introduce a two-player game, the so-called \emph{\iddominance game}, which is inspired by the delayed simulation game for alternating Büchi automata~\cite{FritzW05}.
Given an ACA $\aca{\varphi} = (Q, q_0, \delta, F)$, two strategies~$s$ and $t$ for some process $p_i$, and an input sequence $\gamma \in (2^{\inputs{i}})^\omega$, the \iddominance game determines whether $s$ \iddominates $t$ for $\aca{\varphi}$ on input $\gamma$.
Intuitively, the game proceeds in rounds. At the beginning of each round, a pair $(p,q)$ of states $p, q \in Q$ and the number of the iteration $j \in \mathbb{N}$ is given, where $p$ represents a state that is visited by a run of~$\aca{\varphi}$ induced by $\computation{t}{\gamma}$, while $q$ represents a state that is visited by a run of~$\aca{\varphi}$ induced by $\computation{s}{\gamma}$.
We call $p$ the \emph{alternative state} and $q$ the \emph{dominant state}.
Let $\sigma^s := \computation{s}{\gamma}$ and $\sigma^t := \computation{t}{\gamma}$. The players \dom and \alt, where \dom takes on the role of Player 0, play as follows: 
\textcolor{lipicsGray}{\sffamily\bfseries\upshape\mathversion{bold}{1.}} \alt chooses a set $c \in \delta(p,\sigma^t_j)$.
\textcolor{lipicsGray}{\sffamily\bfseries\upshape\mathversion{bold}{2.}} \dom chooses a set $c' \in \delta(q,\sigma^s_j)$.
\textcolor{lipicsGray}{\sffamily\bfseries\upshape\mathversion{bold}{3.}} \alt chooses a state $q' \in c'$.
\textcolor{lipicsGray}{\sffamily\bfseries\upshape\mathversion{bold}{4.}} \dom chooses a state $p' \in c$.
The starting pair of the next round is then $((p',q'),j+1)$. Starting from $((q_0,q_0),0)$, the players construct an infinite play which determines the winner. 
\dom wins for a play if every rejecting dominant state is matched by a rejecting alternative state eventually.

Both the \iddominant strategy $s$ and the alternative strategy $t$ may control the nondeterministic transitions of $\aca{\varphi}$, while the universal ones are uncontrollable. Since, intuitively, strategy~$t$ is controlled by an opponent when proving that~$s$ \iddominates~$t$, we thus have a change in control for $t$:
for~$s$, \dom controls the existential transitions of $\aca{\varphi}$ and \alt controls the universal ones. For~$t$, \dom controls the universal transitions and \alt controls the existential ones.
Note that the order in which \alt and \dom make their moves is crucial to ensure that \dom wins the game when considering the very same process strategies. By letting \alt move first, \dom is able to mimic -- or \emph{duplicate} -- \alt's moves.
Formally, the \iddominance game is defined as follows:

\enlargethispage{1.5\baselineskip}
\begin{definition}[\iDDominance game]
Let $\aca{\varphi}\! = \!(Q, q_0, \delta, F)$ be an ACA. Based on $\mathcal{A}$, we define the sets $S_\exists = (Q \times Q) \times \mathbb{N}$, $D_\exists = (Q \times Q \times 2^Q) \times \mathbb{N}$, $S_\forall = (Q \times Q \times 2^Q \times 2^Q) \times \mathbb{N}$, and $D_\forall = (Q \times Q \times Q \times 2^Q) \times \mathbb{N}$.
Let ${\sigma, \sigma' \in (2^\variables{i})^\omega\!}$ be infinite sequences. Then, the \emph{\iddominance game $(\aca{\varphi}, \sigma, \sigma')$} is the game $\mathcal{G} = (\arena, \win)$ defined by $\arena = (V,V_0,V_1,v_0,E)$ with $V = S_\exists \cup D_\exists \cup S_\forall \cup D_\forall$, $V_0 = D_\exists \cup D_\forall$, and $V_1 = S_\exists \cup S_\forall$ as well as
\begin{align*}
E &= \{ (((p,q),j),((p,q,c),j) \mid c \in \delta(p,\sigma_j) \} \cup \{ (((p,q,c),j),((p,q,c,c'),j) \mid c' \in \delta(q,\sigma'_j) \} \\
&\cup \{ (((p,q,c,c'),j),((p,q,c,q'),j) \mid q' \in c' \} \cup \{ (((p,q,c,q'),j),((p',q'),j+1) \mid p' \in c \},
\end{align*}
and the winning condition $\win = \{ \rho \in V^\omega \mid \forall j \in \mathbb{N}.~ \projq{\rho_j} \in F \rightarrow \exists j' \geq j.~ \projp{\rho_{j'}} \in F \}$, where $\projp{v} := \projection{\projection{v}{1}}{1}$ and $\projq{v} := \projection{\projection{v}{1}}{2}$, \ie, $\projp{v}$ and $\projq{v}$ map a position $v$ to the alternative state and the dominant state of $v$, respectively.
\end{definition}

We now define the notion of \iddominance based on the \iddominance game. Intuitively, the winner of the game for the computations of two strategies $s$ and~$t$ determines whether or not $s$ \iddominates $t$ on a given input sequence. Similar to remorsefree dominance, we then lift this definition to \iddominant strategies. Formally:

\begin{definition}[\iDDominant Strategy]
Let $\aca{\varphi}$ be an ACA. Let~$s$ and~$t$ be strategies for process $p_i$.
Then,~\emph{$s$ \iddominates~$t$ on input sequence $\gamma \in (2^\inputs{i})^\omega\!$} for $\aca{\varphi}$, denoted~$\ddominatesSequence{s}{t}{\gamma}{\aca{\varphi}}$, if \dom wins the \iddominance game $(\aca{\varphi}, \computation{t}{\gamma}, \computation{s}{\gamma})$.
Strategy~\emph{$s$ \iddominates~$t$} for~$\aca{\varphi}$, denoted $\ddominatesNotation{s}{t}{\aca{\varphi}}$, if $\ddominatesSequence{s}{t}{\gamma}{\aca{\varphi}}$ holds for all input sequences~$\gamma \in (2^\inputs{i})^\omega\!$.
Strategy~\emph{$s$ is \iddominant} for $\aca{\varphi}$ if, for every alternative strategy $t$ for $p_i$, $\ddominatesNotation{s}{t}{\aca{\varphi}}$ holds.
\end{definition}

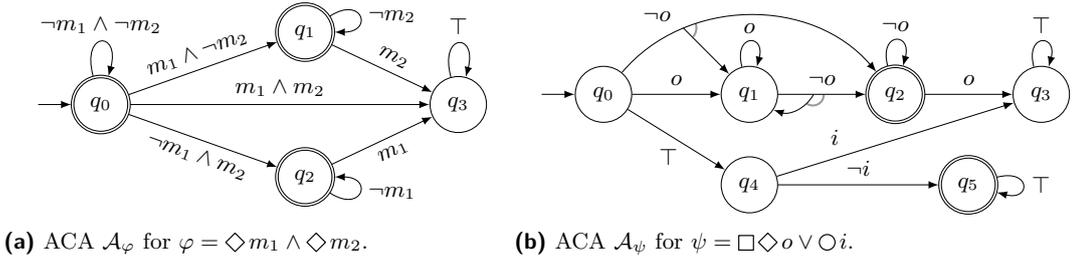
\begin{figure}[t]
\centering
\begin{subfigure}[b]{0.46\textwidth}
\centering
\scalebox{0.96}{
\begin{tikzpicture}[>=latex,shorten >=0pt,auto,->,node distance=1cm,thin,every edge/.style={draw,font=\small},every state/.style={minimum size=2.2em}, initial text =]
	
		\node[state,initial,accepting]	(q0)		at (0,0)		{\small$q_0$};
		\node[state,accepting]			(q1)		at (2.8,1)	{\small$q_1$};
		\node[state,accepting]			(q2)		at (2.8,-1)	{\small$q_2$};
		\node[state]					(q3)		at (4.9,0)	{\small$q_3$};
		
		\path	(q0)		edge[loop above,looseness=8]		node		{$\neg m_1 \land \neg m_2$}		(q0)
						edge	[sloped]			node		{$m_1 \land \neg m_2$}	(q1)
						edge	[sloped,below]	node		{$\neg m_1 \land m_2$}	(q2)
						edge					node		{$m_1 \land m_2$}		(q3)
				(q1)		edge[sloped,pos=0.55]			node		{$m_2$}			(q3)
						edge[loop right,looseness=6,in=0,out=33]		node		{$\neg m_2$}		(q1)
				(q2)		edge[sloped,below,pos=0.55]	node		{$m_1$}			(q3)
						edge[loop right,looseness=6,in=0,out=-33]		node		{$\neg m_1$}		(q2)
				(q3)		edge[loop above,looseness=8]		node		{$\top$}			(q3);
	\end{tikzpicture}}
\caption{ACA $\mathcal{A}_\varphi$ for $\varphi = \protect\Eventually m_1 \land \protect\Eventually m_2$.}\label{fig:ACA_running_example}
\end{subfigure}
\hfill
\begin{subfigure}[b]{0.52\textwidth}
\centering
\scalebox{0.96}{
\begin{tikzpicture}[>=latex,shorten >=0pt,auto,->,node distance=1cm,thin,every edge/.style={draw,font=\small},every state/.style={minimum size=2.2em}, initial text =]

		\def\centerarc[#1](#2)(#3:#4:#5){ \draw[#1] ($(#2)+({#5*cos(#3)},{#5*sin(#3)})$) arc (#3:#4:#5); }
	
		\node[state,initial]		(q0)		at (0,0)		{\small$q_0$};
		\node[state]			(q1)		at (2,0)		{\small$q_1$};
		\node[state,accepting]	(q2)		at (4,0)		{\small$q_2$};
		\node[state]			(q3)		at (6,0)		{\small$q_3$};
		\node[state]			(q4)		at (2,-1.25)	{\small$q_4$};
		\node[state,accepting]	(q5)		at (5,-1.25)	{\small$q_5$};
		
		\node[draw=none]		(a1)		at ($(q1.east)+(0.6,0.12)$) {};
		\node[draw=none]		(a2)		at ($(q0.east)+(0.7,1.04)$) {};

		\centerarc[gray!80,-,thick](a1.south west)(-125:0:0.15)
		\centerarc[gray!80,-,thick](a2.south)(-42:21:0.18)

		\path	(q0)		edge		node		{$o$}		(q1)
						edge		node[below]		{$\top~~$}		(q4)
				(q1)		edge	[loop above,looseness=6]		node		{$o$}		(q1)
				(q2)		edge	[loop above,looseness=6]		node		{$\neg o$}	(q2)
						edge		node		{$o$}		(q3)
				(q3)		edge[loop above,looseness=6]		node		{$\top$}	(q3)
				(q4)		edge		node[pos=0.3]		{$i$}		(q3)
						edge		node		{$\neg i$}	(q5)
				(q5)		edge[loop right,looseness=6]		node		{$\top$}	(q5);
				
		\path	(q1)		edge		node	{$\neg o$}		(q2)
				(a1)		edge[bend left=40]	node		{}	(q1)
				(q0)		edge[bend left=50]	node	[pos=0.25]	{$\neg o$}	(q2)
				(a2.south)		edge		node	{}	(q1);
				
	\end{tikzpicture}}
\caption{ACA $\mathcal{A}_\psi$ for $\psi = \protect\Globally\protect\Eventually o \lor \protect\Next i$.}\label{fig:ACA_interesting}
\end{subfigure}
\caption{Alternating co-Büchi automata $\mathcal{A}_\varphi$ and $\mathcal{A}_\psi$. Universal choices are depicted by connecting the transitions with a gray arc. Rejecting states are marked with double circles.}\label{fig:ACAs_examples}
\end{figure}

As an example for \iddominance, consider the message sending system again. Let $s_i$ be a strategy for process $p_i$ that outputs $m_i$ in the very first step and let $t_i$ be a strategy that waits for $m_{3-i}$ before sendings its own message.
An ACA~$\mathcal{A}_\varphi$ with $\Lang{\mathcal{A}_\varphi} = \Lang{\varphi}$ is depicted in \Cref{fig:ACA_running_example}. Note that~$\mathcal{A}_\varphi$ is deterministic and thus every sequence induces a single run tree with a single branch. 
Hence, for every input sequence $\gamma \in (2^\inputs{1})^\omega$, the moves of both \alt and \dom are uniquely defined by the computations of $t_1$ and $s_1$ on~$\gamma$, respectively.
Therefore, we only provide the state pairs $(p,q)$ of the \iddominance game, not the intermediate tuples.
First, consider an input sequence $\gamma \in (2^\inputs{1})^\omega$ that contains the very first~$m_{2}$ at point in time~$\ell$. Then, the run of $\mathcal{A}_\varphi$ on $\computation{s_1}{\gamma}$ starts in $q_0$, moves to~$q_1$ immediately if $\ell > 0$, stays there up to the occurrence of $m_{2}$ and then moves to~$q_3$, where it stays forever. If $\ell = 0$, then the run moves immediately from $q_0$ to $q_3$. The run of $\computation{t_1}{\gamma}$, in contrast, stays in~$q_0$ until $m_{2}$ occurs, then moves to~$q_2$ and then immediately to~$q_3$, where it stays forever. Thus, we obtain the unique sequence $(q_0,q_0) (q_0,q_1)^{\ell-1} (q_2,q_3) (q_3,q_3)^\omega$ of state pairs in the \iddominance game $(\mathcal{A}_\varphi,\computation{t_1}{\gamma},\computation{s_1}{\gamma})$. The last rejecting alternative state, \ie, a rejecting state induced by $\computation{t_1}{\gamma}$ occurs at point in time~$\ell+1$, namely~$q_2$, while the last rejecting dominant state \ie, a rejecting state induced by $\computation{s_1}{\gamma}$, occurs at point in time~$\ell$, namely~$q_1$. Thus, $\ddominatesSequence{s_1}{t_1}{\gamma}{\mathcal{A}_\varphi}$ holds. In fact, $\ddominatesSequence{s_1}{t'_1}{\gamma}{\mathcal{A}_\varphi}$ holds for \emph{all} alternative strategies~$t'_1$ for such an input sequence $\gamma$ since every strategy $t'_1$ for $p_1$ induces at least $\ell$ visits to rejecting states due to the structure of $\gamma$.
Second, consider an input sequence $\gamma' \in (2^\inputs{1})^\omega$ that does not contain any~$m_{2}$. Then, the run of $\mathcal{A}_\varphi$ on a computation of any strategy $t'_1$ on~$\gamma'$ never reaches~$q_3$ and thus only visits rejecting states. Hence, in particular, every visit to a rejecting state induced by $\computation{s_1}{\gamma'}$ is matched by a visit to a rejecting state induced by $\computation{t'_1}{\gamma'}$ for all strategies $t'_1$. Thus, $\ddominatesSequence{s_1}{t'_1}{\gamma'}{\mathcal{A}_\varphi}$ holds for all alternative strategies $t'_1$ as well.
We can thus conclude that $s_1$ is \iddominant for $\mathcal{A}_\varphi$, meeting our intuition that $s_1$ should be allowed to violate $\varphi$ on input sequences that do not contain any $m_{2}$.
Strategy~$t_1$, in contrast, is remorsefree dominant for $\varphi$ but not \iddominant for~$\mathcal{A}_\varphi$: consider again an input sequence $\gamma \in (2^\inputs{1})^\omega$ that contains the very first $m_{2}$ at point in time~$\ell$. For the \iddominance game $(\mathcal{A}_\varphi,\computation{s_1}{\gamma},\computation{t_1}{\gamma})$, we obtain the following sequence of state pairs: $(q_0,q_0) (q_1,q_0)^{\ell-1} (q_3,q_2) (q_3,q_3)^\omega$. It contains a rejecting dominant state, \ie, a rejecting state induced by $\computation{t_1}{\gamma}$, at point in time $\ell+1$, while the last rejecting alternative state occurs at point in time $\ell$. 
Hence, $t_1$ does not \iddominate~$s_1$, preventing that it is \iddominant to wait for the other process indefinitely.

Next, consider the LTL formula $\psi = \Globally\Eventually o \lor \Next i$, where $i$ is an input variable and $o$ is an output variable. An ACA $\mathcal{A}_\psi$ with $\Lang{\mathcal{A}_\psi} = \Lang{\psi}$ is depicted in \Cref{fig:ACA_interesting}. Note that it has both existential and~universal transitions. Consider a process strategy~$s$ that outputs $o$ in every step. Let~$t$ be some alternative strategy and let $\gamma$ be some input sequence.
Then, \dom encounters an existential choice in state $q_0$ for $s$ in the very first round of the \iddominance game $(\mathcal{A}_\psi,\computation{t}{\gamma},\computation{s}{\gamma})$: it can choose to move to $q_1$ or to $q_4$. If \dom chooses to move to $q_1$, then the only possible successor state in every run of~$\mathcal{A}_\psi$ induced by $\computation{s}{\gamma}$ is $q_1$. Thus, irrespective of \alt's moves, the sequence of dominant states in all consistent initial plays is given by $q_0 q_1^\omega$. Since neither~$q_0$ nor~$q_1$ is rejecting, \dom wins the game.
Therefore, there exists a winning strategy for \dom for the game $(\mathcal{A}_\psi,\computation{t}{\gamma},\computation{s}{\gamma})$ for all $t$ and $\gamma$, namely choosing to move to $q_1$ from $q_0$, and thus $s$ is \iddominant.
Second, consider a strategy~$t$ that does not output~$o$ in the first step but outputs~$o$ in every step afterwards. Let $\gamma$ be an input sequence that does not contain~$i$ at the second point in time.
Then, \dom encounters an existential choice in state~$q_0$ for~$t$ in the very first round of the \iddominance game $(\mathcal{A}_\psi,\computation{s}{\gamma},\computation{t}{\gamma})$. Yet, if \dom chooses the transition from $q_0$ to~$q_4$, then every consistent play will contain infinitely many rejecting dominant states since the structure of $\gamma$ enforces that every consistent play enters $q_5$ in its dominant state in the next round of the game.
Otherwise, \ie, if \dom chooses the universal transition to both~$q_1$ and~$q_2$, then \alt decides which of the states is entered. If \alt chooses $q_2$, then every consistent play visits a rejecting dominant state, namely $q_2$, in the second round of the game. If \alt further chooses to move from $q_0$ to $q_1$ for the alternative strategy $s$, then, as shown above, no rejecting dominant states are visited in a consistent play at all.
Thus, there exists a winning strategy for \alt and therefore $t$ is not \iddominant for $\mathcal{A}_\psi$.

Recall that one of the main weaknesses of bounded dominance is that every strategy, even a non-dominant one, is trivially $n$-dominant if the bound $n$ is chosen too small. 
Every \iddominant strategy, in contrast, is also remorsefree dominant.
The main idea is that a winning strategy $\tau$ of \dom in the \iddominance game defines a run tree of the automaton induced by the \iddominant strategy $s$ such that all branches either visit only finitely many rejecting states or such that all rejecting states are matched eventually with a rejecting state in some branch, which is also defined by $\tau$, of all run trees induced by an alternative strategy. Thus, $s$ either satisfies the specification, or an alternative strategy does not satisfy it either. For the formal proof, we refer the reader to~\Cref{app:ddominance}.

\begin{theorem}\label{lem:ddom_implies_dom}
Let $\varphi$ be an LTL formula.
Let $\mathcal{A}_{\varphi}$ be an ACA with $\Lang{\mathcal{A}_{\varphi}} = \Lang{\varphi}$. Let $s$ be a strategy for process $p_i$. If $s$ is \iddominant for~$\mathcal{A}_{\varphi}$, then $s$ is remorsefree dominant for~$\varphi$.
\end{theorem}

Clearly, the converse does not hold. For instance, a strategy in the message sending system that waits for the other process to send its message first is remorsefree dominant for~$\varphi$ but not \iddominant for the ACA depicted in \Cref{fig:ACA_running_example} as pointed out above.

Given an LTL formula $\varphi$, for remorsefree dominance it holds that if $\varphi$ is realizable, then every strategy that is dominant for $\varphi$ is also winning for~$\varphi$~\cite{DammF14}. This is due to the fact that the winning strategy needs to be taken into account as an alternative strategy for every dominant one, and that remorsefree dominance is solely defined on the satisfaction of the specification. With \Cref{lem:ddom_implies_dom} the same property follows for \iddominance.

\begin{lemma}\label{lem:winning_if_realizable}
Let $\varphi$ be an LTL formula. 
Let $\mathcal{A}_{\varphi}$ be an ACA with $\Lang{\mathcal{A}_{\varphi}} = \Lang{\varphi}$.
If $\varphi$ is realizable, then every \iddominant strategy for $\mathcal{A}_\varphi$ is winning for~$\varphi$ as well.
\end{lemma}

A critical shortcoming of remorsefree dominance is its non-compositionality for liveness properties. This restricts the usage of dominance-based compositional synthesis algorithms to safety specifications, which are in many cases not expressive enough.
\iDdominance, in contrast, is specifically designed to be compositional for more properties:
we identified that a crucial requirement for the compositionality of a process property such as remorsefree dominance or \iddominance is the existence of \emph{bad prefixes} for strategies that do not satisfy the process requirement.
Since remorsefree dominance solely considers the satisfaction of the specification $\varphi$, a bad prefix for a strategy that is not remorsefree dominant boils down to a bad prefix of $\Lang{\varphi}$ and therefore compositionality cannot be guaranteed for liveness properties.
As \iddominance takes the ACA representing $\varphi$ and, in particular, its rejecting states into account, the absence of a bad prefix for $\Lang{\varphi}$ does not necessarily result in the absence of a bad prefix for \iddominance. First, we define such bad prefixes formally:

\begin{definition}[Bad Prefixes for \iDDominance]\label{def:bad_prefixes}
	Let $\mathcal{P}$ be the set of all system processes and all parallel compositions of subsets of system processes. Let~$I_p$ and $O_p$ be the sets of inputs and outputs of $p \in \mathcal{P}$.
		Let $\aca{\varphi}$ be an ACA.
	Then, $\aca{\varphi}$ \emph{ensures bad prefixes for \iddominance} if, for all $p \in \mathcal{P}$ and all strategies~$s$ for~$p$ for which there exists some $\gamma \in (2^{I_p})^\omega$ such that $\ddominatesSequenceNegated{t}{s}{\gamma}{\aca{\varphi}}$ holds for some alternative strategy $t$, there is a finite prefix $\eta \in (2^{I_p \cup O_p})^*$ of $\computation{s}{\gamma}$ such that for all infinite extensions $\sigma\in (2^{I_p \cup O_p})^\omega$ of $\eta$, there is an infinite sequence $\sigma'\in (2^{I_p \cup O_p})^\omega$ with $\sigma \cap I_p = \sigma' \cap I_p$ such that \dom loses the \iddominance game $(\aca{\varphi},\sigma',\sigma)$.
\end{definition}

Intuitively, an ACA that ensures bad prefixes for \iddominance thus guarantees that for every strategy that is not \iddominant, there exists a point in time at which its behavior ultimately prevents \iddominance, irrespective of any future behavior. For more details on bad prefixes for \iddominance and their existence in automata, we refer to~\textcolor{red}{TODO}.
If an ACA ensures bad prefixes for \iddominance, compositionality is then guaranteed:
if the parallel composition of two \iddominant strategies $s_1$ and $s_2$ is not \iddominant, then the behavior of both processes at the last position of the smallest bad prefix reveals which one of them is responsible for \dom losing the game. Note that also both processes can be responsible simultaneously. 
Since there is an alternative strategy $t$ for the composed system for which \dom wins the game, as otherwise $s_1 \pc s_2$ would be \iddominant, the strategy of the process $p_i$ which is responsible for \dom losing the game cannot be \iddominant since there is an alternative strategy, namely~$t$ restricted to the outputs of~$p_i$, that allows \dom to win the game. 
For the formal proof, we refer to~\Cref{app:ddominance}.

\begin{theorem}[Compositionality of \iDDominance]\label{thm:compositonality_ddom}
Let $\aca{\varphi}$ be an ACA that ensures bad prefixes for \iddominance. Let~$s_1$ and $s_2$ be \iddominant strategies for $\aca{\varphi}$ and processes~$p_1$ and~$p_2$, respectively. Then, $s_1 \pc s_2$ is \iddominant for $\aca{\varphi}$ and $p_1 \pc p_2$.
\end{theorem}

From \Cref{lem:ddom_implies_dom,thm:compositonality_ddom} it then follows immediately that the parallel composition of two \iddominant strategies is also remorsefree dominant if the ACA ensures bad prefixes:

\begin{corollary}
Let $\varphi$ be an LTL formula. Let $\mathcal{A}_{\varphi}$ be an ACA with $\Lang{\mathcal{A}_{\varphi}}=\Lang{\varphi}$ that ensures bad prefixes for \iddominance. Let~$s_1$ and~$s_2$ be \iddominant strategies for~$\mathcal{A}_{\varphi}$ and processes~$p_1$ and~$p_2$, respectively. Then, $s_1 \pc s_2$ is remorsefree dominant for $\varphi$ and $p_1 \pc p_2$. {}
\end{corollary}

\enlargethispage{1.5\baselineskip}
With \Cref{lem:winning_if_realizable,thm:compositonality_ddom} we obtain that, given a specification $\varphi$ and an ACA $\mathcal{A}_\varphi$ with $\Lang{\mathcal{A}_\varphi}=\Lang{\varphi}$ that ensures bad prefixes for \iddominance, the parallel composition of \iddominant strategies for $\mathcal{A}_\varphi$ and all processes of a distributed system is winning if $\varphi$ is realizable. Hence, \iddominance can be soundly used for dominance-based compositional synthesis approaches when ensuring the bad prefix criterion.
In the next section, we thus introduce an automaton construction for synthesizing \iddominant strategies.


\section{Synthesizing \texorpdfstring{\iDDominant}{Delay-Dominant} Strategies}\label{sec:automaton_construction}

In this section, we introduce how \iddominant strategies can be synthesized using existing tools for synthesizing winning strategies. We focus on utilizing \emph{bounded synthesis} tools such as BoSy~\cite{FaymonvilleFT17}. Mostly, we use bounded synthesis as a black box procedure throughout this section. Therefore, we do not go into detail here and refer the interested reader to~\cite{FinkbeinerS13,FaymonvilleFRT17}. A crucial observation regarding bounded synthesis that we utilize, however, is that it translates the given specification $\varphi$ into an equivalent universal co-Büchi automaton~$\mathcal{A}_\varphi$ and then derives a strategy such that, for every input sequence, the runs of $\mathcal{A}_\varphi$ induced by the computation of the strategy on the input sequence visit only finitely many rejecting states.

\enlargethispage{\baselineskip}
To synthesize \iddominant strategies instead of winning ones, we can thus use existing bounded synthesis algorithms by replacing the universal co-Büchi automaton $\mathcal{A}_\varphi$ with one encoding \iddominance, \ie, with an automaton $\ddUCA{\varphi}$ such that its runs induced by the computations of a \iddominant strategy on all input sequences visit only finitely many rejecting states. This idea is similar to the approach for synthesizing remorsefree dominant strategies~\cite{DammF14,FinkbeinerP20}.
The automaton for recognizing \iddominant strategies, however, differs inherently from the one for recognizing remorsefree dominant strategies.

\begin{figure}[t]
\centering
\tikzstyle{box} = [draw, text centered, rounded corners=5pt, inner sep=0.2cm, align=center]
\tikzstyle{splitbox} = [box, rectangle split, rectangle split parts=2, inner sep = 0.2cm]
\scalebox{0.9}{
\begin{tikzpicture}[>=latex,shorten >=1pt,auto,node distance=5cm and 5cm,semithick]	
				
				\pic [local bounding box=ACAPhi] at (0,0) {vhsplit={$\mathcal{A}_\varphi$\vphantom{$\ddUCA{\varphi}$}}{ACA}{$2^\variables{i}$\vphantom{$2^\primedOutputs{i}$}}{gray!20}};
				
				\pic [local bounding box=ACAnegPhi] at (0,-1.6) {vhsplit={$\mathcal{A}_{\neg\varphi}$\vphantom{$\ddUCA{\varphi}$}}{ACA}{$2^\variables{i}$\vphantom{$2^\primedOutputs{i}$}}{gray!20}};
				
				\pic [local bounding box=ACADD] at (3.5,-0.8) {vhsplit={$\ddACA{\varphi}$\vphantom{$\ddUCA{\varphi}$}}{ACA}{$2^{\variables{i} \cup \primedOutputs{i}}$}{gray!20}};
				
				\pic [local bounding box=UCADD] at (7.7,-0.8) {vhsplit={$\ddUCAnonProj{\varphi}$\vphantom{$\ddUCA{\varphi}$}}{UCA}{$2^{\variables{i} \cup \primedOutputs{i}}$}{gray!20}};
				
				\pic [local bounding box=UCADDProj] at (12.1,-0.8) {vhsplit={$\ddUCA{\varphi}$}{UCA}{$2^{\variables{i}}$\vphantom{$2^\primedOutputs{i}$}}{gray!20}};
				
				\begin{scope}[on background layer]
					\node[box,fill=gray!20,minimum width=0.85cm] (Phi) at ($(ACAPhi.west) + (-0.86,0)$)	{$\varphi$};
					\node[box,fill=gray!20,minimum width=0.85cm] (negPhi) at ($(ACAnegPhi.west) + (-0.86,0)$)	{$\!\neg\varphi$};
					\node (D) at ($(ACAPhi.east) + (0.8,-0.5)$) {\small Def.~\ref{def:aca_dd}};
				\end{scope}
				
				\draw[very thick,->] (ACAPhi.east) -- ($(ACAPhi.east)+(0.25,0)$) -- ($(ACAPhi.east)+(0.25,-0.8)$) -- (ACADD.west);
				\draw[very thick,->] (ACAnegPhi.east) -- ($(ACAnegPhi.east)+(0.25,0)$) -- ($(ACAnegPhi.east)+(0.25,0.8)$) -- (ACADD.west);
				
				\path	(Phi) 			edge		[very thick, ->]		node {}	(ACAPhi.west)
						(negPhi) 		edge		[very thick, ->]		node {}	(ACAnegPhi.west)
						(ACADD.east) 	edge 	[very thick, ->]		node[above,pos=0.45]	{\small Miyano-} (UCADD.west)
										edge 	[very thick, ->]		node[below,pos=0.45]	{\small Hayashi} (UCADD.west)
						(UCADD.east)		edge 	[very thick, ->]		node[above,pos=0.45]	{\small universal\vphantom{y}} (UCADDProj.west)
										edge 	[very thick, ->]		node[below,pos=0.46]	{\small projection\vphantom{H}} (UCADDProj.west);
		\end{tikzpicture}}
\caption{Overview of the construction of a universal co-Büchi automaton $\ddUCA{\varphi}$ recognizing \iddominant strategies for the alternating co-Büchi automaton $\mathcal{A}_\varphi$ with $\Lang{\mathcal{A}_\varphi} = \Lang{\varphi}$. The lower parts of the boxes list the automaton type (alternating or universal) and the alphabet.}\label{fig:overview}
\end{figure}
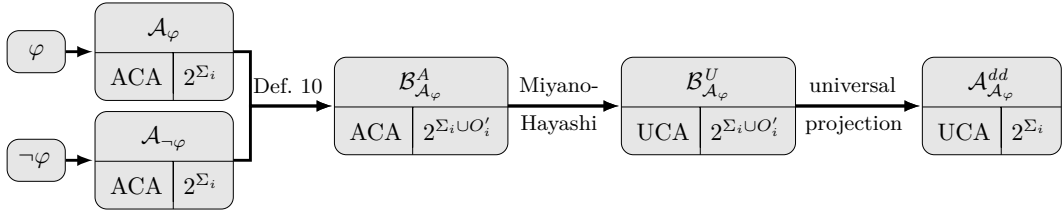

The automaton construction consists of several steps. An overview is given in \Cref{fig:overview}. Since \iddominance is not defined on the LTL specification $\varphi$ itself but on an equivalent alternating co-Büchi automaton, we first translate $\varphi$ into an alternating co-Büchi automaton~$\mathcal{A}_\varphi$ with $\Lang{\mathcal{A}_\varphi} = \Lang{\varphi}$. For this, we utilize well-known algorithms for translating LTL formulas into equivalent alternating Büchi automata as well as the duality of the Büchi and co-Büchi acceptance condition and of nondeterministic and universal branching. More details on the translation of LTL formulas into alternating co-Büchi automata are provided in~\Cref{app:preliminaries}. Similarly, we construct an alternating co-Büchi automaton $\mathcal{A}_{\neg\varphi}$ with $\Lang{\mathcal{A}_{\neg\varphi}} = \Lang{\neg\varphi}$ from $\neg\varphi$.
The centerpiece of the construction is an alternating co-Büchi automaton~$\ddACA{\varphi}$ constructed from~$\mathcal{A}_\varphi$ and~$\mathcal{A}_{\neg\varphi}$ that recognizes whether $\ddominatesSequence{s}{t}{\gamma}{\mathcal{A}_\varphi}$ holds for $\mathcal{A}_\varphi$, input sequence $\gamma \in (2^\inputs{i})^\omega$ and strategies~$s$ and~$t$ for process~$p_i$.
The alternating automaton~$\ddACA{\varphi}$ is then translated into an equivalent universal co-Büchi automaton~$\ddUCAnonProj{\varphi}$, for example with the Miyano-Hayashi algorithm~\cite{MiyanoH84}.
Lastly, we translate~$\ddUCAnonProj{\varphi}$ into a universal co-Büchi automaton that accounts for requiring a strategy~$s$ to \iddominate \emph{all} other strategies $t$ and not only a particular one utilizing universal projection. In the remainder of this section, we describe all steps of the construction in detail and prove their correctness.

\subsection{Construction of the ACA \texorpdfstring{\boldmath $\ddACA{\varphi}$}{B{A}}}

\enlargethispage{\baselineskip}
From the two ACAs $\mathcal{A}_\varphi$ and $\mathcal{A}_{\neg\varphi}$, we construct an alternating co-Büchi automaton $\ddACA{\varphi}$ that recognizes whether $\ddominatesSequence{s}{t}{\gamma}{\mathcal{A}_\varphi}$ holds for $\mathcal{A}_\varphi$, input sequence $\gamma \in (2^\inputs{i})^\omega$ and process strategies~$s$ and $t$ for process~$p_i$.
The construction relies on the observation that $\ddominatesSequence{s}{t}{\gamma}{\mathcal{A}_\varphi}$ holds if, and only if, either~(i)~$\computation{t}{\gamma} \not\models \varphi$ holds or~(ii)~we have $\ddominatesSequence{s}{t}{\gamma}{\mathcal{A}_\varphi}$ and every initial play of the \iddominance game that is consistent with the winning strategy of \dom visits only finitely many rejecting dominant states. The proof of this observation is provided in~\Cref{app:automaton_construction}.
Therefore, the automaton $\ddACA{\varphi}$ consists of two parts, one accounting for~(i)~and one accounting for~(ii), and guesses nondeterministically in the initial state which part is entered. The ACA $\mathcal{A}_{\neg\varphi}$ with $\Lang{\mathcal{A}_{\neg\varphi}}=\Lang{\neg\varphi}$ accounts for~(i). For~(ii), we intuitively build the product of two copies of the ACA $\mathcal{A}_\varphi$ with $\Lang{\mathcal{A}_{\varphi}}=\Lang{\varphi}$, one for each of the considered process strategies $s$ and $t$. Note that similar to the change of control for $t$ in the \iddominance game, we consider the \emph{dual} transition function of $\mathcal{A}_\varphi$, \ie, the one where conjunctions and disjunctions are swapped, for the copy of $\mathcal{A}_\varphi$ for $t$.
We keep track of whether we encountered a situation in which a rejecting state was visited for $s$ while it was not for $t$. This allows for defining the set of rejecting states.

Note that we need to allow for differentiating valuations of output variables computed by~$s$ and $t$ on the same input sequence. Therefore, we extend the alphabet of $\ddACA{\varphi}$: in addition to the set $\variables{i}$ of variables of process $p_i$, which contains input variables $\inputs{i}$ and output variables~$\outputs{i}$, we consider the set $\primedOutputs{i} := \{ o' \mid o \in \outputs{i} \}$ of \emph{primed output variables} of $p_i$, where every output variable is marked with a prime to obtain a fresh symbol. The set $\primedVariables{i}$ of \emph{primed variables} of~$p_i$ is then given by $\primedVariables{i} := \inputs{i} \cup \primedOutputs{i}$. Intuitively, the output variables $\outputs{i}$ depict the behavior of the \iddominant strategy $s$, while the primed output variables $\primedOutputs{i}$ depict the behavior of the alternative strategy $t$. The alphabet of~$\ddACA{\varphi}$ is then given by $2^{\variables{i} \cup \primedOutputs{i}}$. This is equivalent to $2^{\variables{i} \cup \primedVariables{i}}$ since the input variables are never primed to ensure that we consider the same input sequence for both strategies.
In the following, we use the functions $\primeSequenceName: \variables{i} \rightarrow \primedVariables{i}$ and $\unprimeName: \primedVariables{i} \rightarrow \variables{i}$ to switch between primed variables and normal ones: given a valuation $a \in \variables{i}$ of variables, $\primeSequence{a}$ replaces every output variable $o\in\outputs{i}$ occurring in $a$ with its primed version $o'$. For a valuation $a \in \primedVariables{i}$, $\unprime{a}$ replaces every primed output variable $o'\in\primedOutputs{i}$ occurring in $a$ with its normal unprimed version $o$. We extend $\primeSequenceName$ and $\unprimeName$ to finite and infinite sequences as usual.
The ACA $\ddACA{\varphi}$ is then constructed as follows:

\begin{definition}\label{def:aca_dd}
Let $\varphi$ be an LTL formula over alphabet $2^{\variables{i}}$.
Let $\mathcal{A}_{\varphi} = (Q, q_0, \delta, F)$ be an ACA with $\Lang{\varphi} = \Lang{\mathcal{A}_\varphi}$.
Let $\mathcal{A}_{\neg\varphi} = (Q^c, q^c_0, \delta^c, F^c)$ be an ACA with $\Lang{\neg\varphi} = \Lang{\mathcal{A}_{\neg\varphi}}$.
We construct the ACA $\ddACA{\varphi} = (\ddACAComponent{Q}, \ddACAComponent{Q}_0, \ddACAComponent{\delta}, \ddACAComponent{F})$ with alphabet $2^{\variables{i} \cup \primedOutputs{i}}\!$ as follows.
\begin{itemize}
\item $\ddACAComponent{Q} := (Q \times Q \times \{\top,\bot\}) \cup Q^c$
\item $\ddACAComponent{Q}_0 \coloneq (q_0,q_0,\top)$
\item $\ddACAComponent{F} \coloneq (Q \times Q \times \{\bot\}) \cup F^c$
\item $\ddACAComponent{\delta}: ((Q \times Q \times \{\top,\bot\}) \cup Q^c) \times 2^{\variables{i} \cup \primedOutputs{i}} \rightarrow (Q \times Q \times \{\top,\bot\}) \cup Q^c$ with
\begin{align*}
\ddACAComponent{\delta}(q_c,\tilde{\iota}) &\coloneq \delta^c(q_c,\iota') ~~~\text{for $q_c \in Q^c$}\\
\ddACAComponent{\delta}((q_0,q_0,\top),\tilde{\iota}) &\coloneq \delta^c(q_0,\iota') \lor \bigwedge_{c \in \delta(q_0,\iota')} \bigvee_{c' \in \delta(q_0,\iota)} \bigwedge_{q' \in c'} \bigvee_{p'\in c} \successor{p'}{q'}{\top}\\
\ddACAComponent{\delta}((p,q,m),\tilde{\iota}) &\coloneq \bigwedge_{c \in \delta(p,\iota')} \bigvee_{c' \in \delta(q,\iota)} \bigwedge_{q' \in c'} \bigvee_{p'\in c} \successor{p'}{q'}{m}
\end{align*}
where $\iota \coloneq \tilde{\iota} \cap \variables{i}$, $\iota' \coloneq \unprime{\tilde{\iota} \cap \primedVariables{i}}$, and $\vartheta: (Q \times Q \times \{\top,\bot\}) \rightarrow Q \times Q \times \{\top,\bot\}$ with \[ \successor{p}{q}{m} \coloneq \begin{cases}
(p,q,\bot) & \text{if $p \not\in F$, $q\in F$, and $m = \top$}\\
(p,q,\bot) & \text{if $p \not\in F$ and $m = \bot$}\\
(p,q,\top) & \text{otherwise}\\
\end{cases} \]
\end{itemize}
\end{definition}

Note that $\ddACA{\varphi}$ indeed consists of two parts: the one defined by states of the form $(p,q,m)$, and the one defined by the states of $\mathcal{A}_{\neg\varphi}$. By definition of $\ddACAComponent{\delta}$, these parts are only connected in the initial state of $\ddACA{\varphi}$, where a nondeterministic transition to the respective successors in both parts ensures that choosing nondeterministically whether (i) or (ii) will be satisfied is possible.
For states of the form $(p,q,m)$, the mark $m \in \{\top,\bot\}$ determines whether there are \emph{pending visits to rejecting states} in the copy of $\mathcal{A}_\varphi$ for the dominant strategy, \ie, the second component $q$ of $(p,q,m)$. A pending visit to a rejecting state is one that is not yet matched by a visit to a rejecting state in the copy of $\mathcal{A}_\varphi$ for the alternative strategy. 
Thus,~$\vartheta$ defines that if a visit to a rejecting dominant state, that is not immediately matched with a rejecting alternative state, is encountered, the mark is set to $\bot$. As long as no rejecting alternative state is visited, the mark stays set to $\bot$. If a matching rejecting alternative state occurs, however, the mark is reset to $\top$.
States of $\ddACA{\varphi}$ marked with $\bot$ are then defined to be rejecting states, ensuring that a visit to a rejecting dominant state is not pending forever. {}

The ACA $\ddACA{\varphi}$ constructed from ACAs $\mathcal{A}_{\varphi}$ and $\mathcal{A}_{\neg\varphi}$ according to \Cref{def:aca_dd} is sound and complete in the sense that it recognizes whether or not a strategy~$s$ \iddominates another strategy~$t$ on an input sequence $\gamma \in (2^\inputs{i})^\omega$. That is, $\ddACA{\varphi}$ accepts the infinite word $\computation{s}{\gamma} \cup \primeSequence{\computation{t}{\gamma}\cap\outputs{i}}$ if, and only if, $\ddominatesSequence{s}{t}{\gamma}{\mathcal{A}_\varphi}$ holds for~$\mathcal{A}_{\varphi}$.
The main idea is that a run tree of~$\ddACA{\varphi}$ can be translated into a strategy for \dom in the \iddominance game and vice versa since, by construction, both define the existential choices in $\mathcal{A}_\varphi$ for~$s$ and the universal choices in $\mathcal{A}_\varphi$ for~$t$. 
Thus, for a run tree of $\ddACA{\varphi}$ whose branches all visit only finitely many rejecting states, there exists a strategy for \dom in the \iddominance game that ensures that for all consistent plays either $\computation{t}{\gamma}\not\models\varphi$ holds or, by construction of~$\vartheta$ and~$\ddACAComponent{\delta}$, every rejecting dominant state is matched by a rejecting alternative state eventually.
Similarly, a winning strategy for \dom can be translated into a run tree $r$ of $\ddACA{\varphi}$. If $\computation{t}{\gamma}\models\varphi$ holds, then $r$ visits only finitely many rejecting states since only finitely many rejecting dominant states are visited. If $\computation{t}{\gamma}\not\models\varphi$ holds, then there exists a run tree, namely one entering the part of $\ddACA{\varphi}$ that coincides with $\mathcal{A}_{\neg\varphi}$, whose branches all visit only finitely many rejecting states.
The proof is given in~\Cref{app:automaton_construction}.

\begin{lemma}\label{lem:soundness_completeness_sequence}
Let $\varphi$ be an LTL formula.
Let $\mathcal{A}_{\varphi}$ and $\mathcal{A}_{\neg\varphi}$ be ACAs with $\Lang{\varphi} = \Lang{\mathcal{A}_\varphi}$ and $\Lang{\neg\varphi} = \Lang{\mathcal{A}_{\neg\varphi}}$. Let $\ddACA{\varphi}$ be the ACA constructed from~$\mathcal{A}_{\varphi}$ and $\mathcal{A}_{\neg\varphi}$ according to \Cref{def:aca_dd}. Let $s$ and $t$ be strategies for process $p_i$. Let $\gamma \in (2^\inputs{i})^\omega\!$. Let $\sigma \in (2^{\variables{i} \cup \primedOutputs{i}})^\omega\!$ with $\sigma := \computation{s}{\gamma} \cup \primeSequence{\computation{t}{\gamma}\cap\outputs{i}}$. Then, $\ddACA{\varphi}$ accepts $\sigma$ if, and only if, $\ddominatesSequence{s}{t}{\gamma}{\mathcal{A}_\varphi}$ holds.
\end{lemma}

Thus, $\ddACA{\varphi}$ determines whether or not a strategy~$s$ \iddominates a strategy~$t$.
However,~$\ddACA{\varphi}$ cannot directly be used for synthesizing \iddominant strategies since~(i)~$\ddACA{\varphi}$ is an alternating automaton, while we require a universal automaton for bounded synthesis, and~(ii)~$\ddACA{\varphi}$ considers \emph{one particular} alternative strategy $t$. For recognizing \iddominance, we need to consider \emph{all} alternative strategies, though.
Thus, we describe in the remainder of this section how $\ddACA{\varphi}$ can be translated into a UCA for bounded synthesis.


\subsection{Construction of the UCA \texorpdfstring{\boldmath $\ddUCA{\varphi}$}{A{dd}}}

Next, we translate the ACA~$\ddACA{\varphi}$ constructed in the previous subsection to a UCA~$\ddUCA{\varphi}$ that can be used for synthesizing \iddominant strategies. As outlined before, we need to~(i)~translate~$\ddACA{\varphi}$ into a UCA, and~(ii)~ensure that the automaton considers \emph{all} alternative strategies instead of a particular one. Thus, we proceed in two steps.
First, we translate~$\ddACA{\varphi}$ into an equivalent UCA~$\ddUCAnonProj{\varphi}$. We utilize the Miyano-Hayashi algorithm~\cite{MiyanoH84} for translating ABAs into NBAs. Since we are considering co-Büchi automata instead of Büchi automata, we further make use of the duality of nondeterministic and universal branching and the Büchi and co-Büchi acceptance conditions. 
The translation introduces an exponential blow-up in the number of states.
For the full construction, we refer to~\Cref{app:automaton_construction}.

\begin{lemma}\label{thm:miyano-hayashi_universal}
Let $\mathcal{A}$ be an alternating co-Büchi automaton with $m$ states. There exists a universal co-Büchi automaton $\mathcal{B}$ with $\mathcal{O}(2^m)$ states such that $\mathcal{L}(\mathcal{A}) = \mathcal{L}(\mathcal{B})$ holds.
\end{lemma}

Next, we construct the desired universal co-Büchi automaton $\ddUCA{\varphi}$ that recognizes \iddominant strategies for $\mathcal{A}_\varphi$. For this sake, we need to adapt $\ddUCAnonProj{\varphi}$ to consider \emph{all} alternative strategies instead of a particular one.
Similar to the automaton construction for synthesizing remorsefree dominant strategies~\cite{DammF14,FinkbeinerP20}, we utilize \emph{universal projection}:

\begin{definition}[Universal Projection]
Let $\mathcal{A} = (Q, Q_0, \delta, F)$ be a UCA over alphabet~$\Sigma$ and let $X \subset \Sigma$. The \emph{universal projection of $\mathcal{A}$ to $X$} is the UCA $\pi_X(\mathcal{A}) = (Q, Q_0, \pi_X(\delta), F)$ over alphabet~$X$, where $\pi_X(\delta) = \{ (q, a, q') \in Q \times 2^X \times Q \mid \exists b \in 2^{\Sigma \setminus X}\!.~ (q, a \cup b, q') \in \delta \}$.
\end{definition}

The projected automaton $\pi_X(\mathcal{A})$ for a UCA $\mathcal{A}$ over $\Sigma$ and a set $X \subset \Sigma$ contains the transitions of $\mathcal{A}$ for \emph{all} possible valuations of the variables in $\Sigma \setminus X$. 
Hence, for a sequence $\sigma \in (2^X)^\omega$, all runs of $\mathcal{A}$ on sequences extending~$\sigma$ with some valuation of the variables in~$\Sigma\setminus X$ are also runs of $\pi_X(\mathcal{A})$.
Since both $\mathcal{A}$ and $\pi_X(\mathcal{A})$ are universal automata,~$\pi_X(\mathcal{A})$ thus accepts a sequence $\sigma \in (2^X)^\omega$ if, and only if, $\mathcal{A}$ accepts all sequences extending~$\sigma$ with some valuation of the variables in $\Sigma \setminus X$. The proof is given in~\Cref{app:automaton_construction}.

\begin{lemma}\label{lem:universal_projection}
Let $\mathcal{A}$ be a UCA over alphabet $\Sigma$ and let $X \subset \Sigma$. Let~$\sigma \in (2^X)^\omega$. Then,~$\pi_X(\mathcal{A})$ accepts $\sigma$ if, and only if $\mathcal{A}$ accepts all $\sigma' \in (2^\Sigma)^\omega$ with $\sigma' \cap X = \sigma$.
\end{lemma}

We utilize this property to obtain a universal co-Büchi automaton~$\ddUCA{\varphi}$ from $\ddUCAnonProj{\varphi}$ that considers \emph{all} possible alternative strategies instead of only a particular one: we project to the unprimed variables of $p_i$, \ie, to $\variables{i}$, thereby quantifying universally over the alternative strategies. We thus obtain a UCA that recognizes \iddominant strategies as follows:

\begin{definition}[\iDDominance Automaton]\label{def:UCA_construction_delayed_dominance}
Let $\varphi$ be an LTL formula. Let $\mathcal{A}_\varphi$, $\mathcal{A}_{\neg\varphi}$ be ACAs with $\Lang{\mathcal{A}_\varphi} = \Lang{\varphi}$, $\Lang{\mathcal{A}_{\neg\varphi}} = \Lang{\neg\varphi}$. Let $\ddACA{\varphi}$ be the ACA constructed from~$\mathcal{A}_\varphi$ and~$\mathcal{A}_{\neg\varphi}$ according to \Cref{def:aca_dd}. Let $\ddUCAnonProj{\varphi}$ be a UCA with $\mathcal{L}(\ddACA{\varphi}) = \mathcal{L}(\ddUCAnonProj{\varphi})$. 
The \emph{\iddominance UCA} $\ddUCA{\varphi}$ for $\mathcal{A}_\varphi$ and process $p_i$ is then given by $\ddUCA{\varphi} \coloneq \pi_{\variables{i}}(\ddUCAnonProj{\varphi})$.
\end{definition}

Utilizing the previous results, we can now show soundness and completeness of the \iddominance universal co-Büchi automaton $\ddUCA{\varphi}$:
from \Cref{lem:soundness_completeness_sequence}, we know that $\ddACA{\varphi}$ recognizes whether or not a strategy $s$ for a process $p_i$ \iddominates another strategy~$t$ for $p_i$ for~$\mathcal{A}_\varphi$ on an input sequence $\gamma \in (2^\inputs{i})^\omega$. By \Cref{thm:miyano-hayashi_universal}, we have $\Lang{\ddUCAnonProj{\varphi}} = \Lang{\ddACA{\varphi}}$. With the definition of the \iddominance UCA, namely $\ddUCA{\varphi} \coloneq \pi_{\variables{i}}(\ddUCAnonProj{\varphi})$, as well as with \Cref{lem:universal_projection}, it then follows that $\ddUCA{\varphi}$ accepts $\computation{s}{\gamma}$ for all input sequences $\gamma \in (2^\inputs{i})^\omega$ if, and only if, $s$ is \iddominant for $\mathcal{A}_\varphi$. For the formal proof, we refer to~\Cref{app:automaton_construction}.

\begin{theorem}[Soundness and Completeness]\label{thm:soundness_completeness_ddUCA}
Let $\varphi$ be an LTL formula and let $\mathcal{A}_\varphi$ be an ACA with $\mathcal{L}(\varphi) = \mathcal{L}(\mathcal{A}_\varphi)$. Let $\ddUCA{\varphi}$ be the \iddominance UCA for $\mathcal{A}_\varphi$ as constructed in \Cref{def:UCA_construction_delayed_dominance}. Let $s$ be a process strategy for process~$p_i$. Then $\ddUCA{\varphi}$ accepts $\computation{s}{\gamma}$ for all input sequences $\gamma \in (2^\inputs{i})^\omega$, if, and only if $s$ is \iddominant for~$\mathcal{A}_\varphi$ and $p_i$.
\end{theorem}

Furthermore, $\ddUCA{\varphi}$ is of convenient size: for an LTL formula $\varphi$, there is an ACA $\mathcal{A}_\varphi$ with $\Lang{\mathcal{A}_\varphi} = \Lang{\varphi}$ such that $\ddUCA{\varphi}$ constructed from~$\mathcal{A}_\varphi$ is of exponential size in the squared length of the formula $\varphi$.
This follows from \Cref{thm:miyano-hayashi_universal} and from the facts that~(i)~$\mathcal{A}_\varphi$ and~$\mathcal{A}_{\neg\varphi}$ both are of linear size in the length of the LTL formula $\varphi$, and~(ii)~universal projection preserves the automaton size.
The proof is given in~\Cref{app:automaton_construction}.

\begin{lemma}\label{thm:automaton_size}
Let $\varphi$ be an LTL formula and let $s$ be a strategy for process $p_i$. There is an ACA $\mathcal{A}_\varphi$ of size $\mathcal{O}(|\varphi|)$ with $\Lang{\mathcal{A}_\varphi} = \Lang{\varphi}$ and a UCA $\ddUCA{\varphi}$ of size $\mathcal{O}(2^{|\varphi|^2})$ such that~$\ddUCA{\varphi}$ accepts $\computation{s}{\gamma}$ for all $\gamma \in (2^\inputs{i})^\omega$, if, and only if, $s$ is \iddominant for~$\mathcal{A}_\varphi$ and $p_i$.
\end{lemma}

Since the automaton construction described in this section is sound and complete, the UCA~$\ddUCA{\varphi}$ can be used for synthesizing \iddominant strategies.
In fact, it immediately enables utilizing existing bounded synthesis tools for the synthesis of \iddominant strategies by replacing the UCA recognizing winning strategies with $\ddUCA{\varphi}$.

Note that, similar as for the UCA recognizing remorsefree dominance~\cite{DammF14}, $\ddUCA{\varphi}$ can be translated into a nondeterministic parity tree automaton with an exponential number of colors and a doubly-exponential number of states in the squared length of the formula. Synthesizing \iddominant strategies thus reduces to checking tree automata emptiness and, if the automaton is non-empty, to extracting a Moore machine representing a process strategy from an accepted tree.
This can be done in exponential time in the number of colors and in polynomial time in the number of states~\cite{Jurdzinski00}. With \Cref{thm:automaton_size}, a doubly-exponential complexity for synthesizing \iddominant strategies thus follows:

\begin{theorem}
Let $\varphi$ be an LTL formula and let $\mathcal{A}_\varphi$ be an ACA with $\Lang{\mathcal{A}_\varphi} = \Lang{\varphi}$. If there exists a \iddominant strategy for $\mathcal{A}_\varphi$, then it can be computed in \emph{2EXPTIME}.
\end{theorem}

It is well-known that synthesizing winning strategies is 2EXPTIME-complete~\cite{PnueliR89}. 
Since there exists a UCA of exponential size in the length of the formula which recognizes remorsefree dominant strategies, dominant strategies can also be synthesized in 2EXPTIME~\cite{DammF14}.
Syn\-the\-si\-zing \iddominant strategies rather than winning or remorsefree dominant ones thus does not introduce any overhead, while it allows for a simple compositional synthesis approach for distributed systems for many safety and liveness specifications.


\section{Compositional Synthesis with \texorpdfstring{\iDDominant}{Delay-Dominant} Strategies}

In this section, we describe a compositional synthesis approach that utilizes \iddominant strategies.
We extend the algorithm described in~\cite{DammF14} from safety specifications to general properties by synthesizing \iddominant strategies instead of remorsefree dominant ones.
Hence, given a distributed architecture and an LTL specification $\varphi$, the compositional synthesis~algorithm proceeds in four steps.
First, $\varphi$ is translated into an equivalent ACA~$\mathcal{A}_\varphi$ using standard algorithms.
Second, for each system process $p_i$, we construct the UCA~$\ddUCA{\varphi}$ that recognizes \iddominant strategies for $\varphi$ and $p_i$ as described in \Cref{sec:automaton_construction}. Note that although the initial automaton $\mathcal{A}_\varphi$ is the same for every process~$p_i$, the UCAs recognizing \iddominant strategies differ: since the processes have different sets of output variables, already the alphabets of the intermediate ACA $\ddACA{\varphi}$ differ for different processes.
Third, a \iddominant strategy $s_i$ is synthesized for each process $p_i$ from the respective UCA~$\ddUCA{\varphi}$ with bounded synthesis.
Lastly, the strategies $s_1, \dots, s_n$ are composed according to the definition of the parallel composition of Moore machines (see \Cref{sec:preliminaries}) into a single strategy~$s$ for the whole distributed system. By \Cref{thm:compositonality_ddom}, the composed strategy $s$ is \iddominant for~$\mathcal{A}_\varphi$ and the whole system if~$\mathcal{A}_\varphi$ ensures bad prefixes for \iddominance. If $\varphi$ is realizable, then, by \Cref{lem:winning_if_realizable}, strategy $s$ is guaranteed to be winning for~$\varphi$.

Note that even for realizable LTL formulas $\varphi$, there does not necessarily exist a \iddominant strategy since \iddominance is not solely defined on the satisfaction of $\varphi$ but on the \emph{structure} of an equivalent ACA~$\mathcal{A}_\varphi$. In certain cases,~$\mathcal{A}_\varphi$ can thus \enquote{punish}~the~\iddominant strategy by introducing rejecting states at clever positions that do not influence acceptance but \iddominance, preventing the existence of a \iddominant strategy.
However, we experienced that an ACA $\mathcal{A}_\varphi$ constructed with standard algorithms from an LTL formula $\varphi$ does not punish \iddominant strategies since $\mathcal{A}_\varphi$ thoroughly follows the structure of $\varphi$ and thus oftentimes does not contain unnecessary rejecting states.
Furthermore, such an ACA oftentimes ensure bad prefixes for \iddominance: in~\textcolor{red}{TODO}, we discuss under which circumstances the bad prefix property is not satisfied and identify critical structures in co-Büchi automata. When constructing ACAs with standard algorithms from LTL formulas, such structures rarely~--~if ever~--~exist.
Simple optimizations like removing rejecting states that do not lie in a cycle from the set of rejecting states have a positive impact on both the existence of \iddominant strategies and on ensuring bad prefixes: such states cannot be visited infinitely often and thus removing them from the set of rejecting states does not alter the language. Nevertheless, rejecting states can enforce non-\iddominance and thus removing unnecessary rejecting states can result in more strategies being \iddominant. Note that with this optimization it, for instance, immediately follows that for safety properties the parallel composition of \iddominant strategies is \iddominant.
Thus, we experienced that for an ACA $\mathcal{A}_\varphi$ constructed from an LTL formula $\varphi$ with standard algorithms it holds in many cases that~(i)~if $\varphi$ allows for a remorsefree dominant strategy, then $\mathcal{A}_\varphi$ allows for an \iddominant strategy, and~(ii)~the parallel composition of \iddominance strategies for~$\mathcal{A}_\varphi$ is \iddominant as well.
Therefore, the compositional synthesis algorithm presented in this section is indeed applicable for many LTL formulas. {}


\enlargethispage{1.5\baselineskip}
\section{Conclusion}

We have presented a new winning condition for process strategies, \iddominance, that allows a strategy to violate a given specification in certain situations. In contrast to the classical notion of winning, \iddominance can thus be used for individually synthesizing strategies for the processes in a distributed system in many cases, therefore enabling a simple compositional synthesis approach.
\iDdominance builds upon remorsefree dominance, where a strategy is allowed to violate the specification as long as no other strategy would have satisfied it in the same situation. 
However, remorsefree dominance is only compositional for safety properties. For liveness properties, the parallel composition of dominant strategies is not necessarily dominant. 
This restricts the use of dominance-based compositional synthesis algorithms to safety specifications, which are often not expressive enough.
\iDdominance, in contrast, is specifically designed to be compositional for more properties.
We have introduced a game-based definition of \iddominance as well as a criterion such that, if the criterion is satisfied, compositionality of \iddominance is guaranteed; both for safety and liveness properties.
Furthermore, every \iddominant strategy is remorsefree dominant, and, for realizable system specifications, the parallel composition of \iddominant strategies for all system processes is guaranteed to be winning for the whole system if the criterion is satisfied.
Hence, \iddominance is a suitable notion for compositional synthesis algorithms. 
We have introduced an automaton construction for recognizing \iddominant strategies. The resulting universal co-Büchi automaton can immediately be used to synthesize \iddominant strategies utilizing existing bounded synthesis approaches. The automaton is of single-exponential size in the squared length of the specification. Thus, synthesizing \iddominant strategies is, as for winning and remorsefree ones, in 2EXPTIME. {}


\bibliography{bib}

\appendix
\section{From LTL to co-Büchi Automata}\label{app:preliminaries}

In this section, we describe how co-Büchi automata, in particular ACAs und UCAs, can be constructed from an LTL formula. We build upon existing constructions for alternating Büchi automata and nondeterministic Büchi automata, respectively.

\subparagraph*{Alternating Automata.}
For an LTL formula $\varphi$, there exists an \emph{alternating Büchi automaton}~$\mathcal{B}_\varphi$ with $\mathcal{O}(|\varphi|)$ states such that $\Lang{\varphi} = \Lang{\mathcal{B}_\varphi}$ holds~\cite{MullerSS88}.
The Büchi and co-Büchi acceptance conditions are dual. Moreover, the duality of nondeterministic and universal branching in $\omega$-automata is well known.
Therefore, we can lift the LTL to ABA translation to an LTL to \emph{alternating co-Büchi automata} translation with the same automaton size by making use of these dualities: we first construct an ABA~$\mathcal{B}_\varphi$ for the \emph{negated formula} $\neg\varphi$. From $\mathcal{B}_\varphi$, we then construct an alternating co-Büchi automaton $\mathcal{A}_\varphi$ by replacing every conjunction in $\mathcal{B}_\varphi$'s transition function with a disjunction and vice versa. Moreover, we interpret the accepting states of $\mathcal{B}_\varphi$ as rejecting states of $\mathcal{A}_\varphi$. Since the number of states does not change with respect to $\mathcal{B}_\varphi$, the resulting automaton $\mathcal{A}_\varphi$ has, $\mathcal{O}(|\varphi|)$ states. Moreover, due to the duality and due to constructing $\mathcal{B}_\varphi$ from $\neg \varphi$, we obtain $\Lang{\varphi} = \Lang{\mathcal{A}_\varphi}$.

\subparagraph*{Universal Automata.}
For an LTL formula $\varphi$, there exists a nondeterministic Büchi automaton~$\mathcal{B}_\varphi$ with $2^{\mathcal{O}(|\varphi|)}$ states such that $\Lang{\varphi} = \Lang{\mathcal{B}_\varphi}$ holds~\cite{KupfermanV05}. 
Similar to the construction for alternating automata, we make use of the duality of the Büchi and co-Büchi acceptance condition as well as of nondeterministic and universal branching: we first construct an NBA~$\mathcal{B}_\varphi$ for the \emph{negated formula} $\neg\varphi$. From~$\mathcal{B}_\varphi$, we then obtain a universal co-Büchi automaton $\mathcal{A}_\varphi$ by interpreting nondeterministic transitions as universal ones and accepting states as rejecting ones.
Then, $\mathcal{A}_\varphi$ has $2^{\mathcal{O}(|\varphi|)}$ states and $\Lang{\varphi} = \Lang{\mathcal{A}_\varphi}$ holds.


\section{\texorpdfstring{\iDDominant}{Delay-Dominant} Strategies}\label{app:ddominance}

In this section, we present the proofs of the lemmas and theorems of \Cref{sec:ddominance}.
For this sake, we first introduce and prove some general properties of \iddominant strategies that we will use in the proofs of \Cref{sec:ddominance}.
First, note that the definition of \iddominance allows for self-dominance. That is, a strategy $s$ \iddominates itself. This is due to the definition of the \iddominance game, in particular, to the order in which \alt and \dom make their moves:

\begin{lemma}\label{lem:ddom_self_dominance}
	Let $\aca{\varphi} = (Q, q_0, \delta, F)$ be an ACA. Let $s$ be a strategy for process $p_i$ and let $\gamma \in (2^\inputs{i})^\omega$ be some input sequence. Then, $\ddominatesSequence{s}{s}{\gamma}{\aca{\varphi}}$ holds.
\end{lemma}
\begin{proof}
	We construct a winning strategy $\tau$ for \dom in the \iddominance game $(\aca{\varphi},\computation{s}{\gamma},\computation{s}{\gamma})$ by mimicking the respective moves of \alt. Since \alt moves first by construction of the game, this is always possible.
	For every initial play $\rho$ that is consistent with $\tau$, we then have $\projq{\rho_j} \in F \rightarrow \projp{\rho_j} \in F$ for all points in time $j \in \mathbb{N}$ and thus, in particular, $\rho \in \win$ holds. Thus, $\tau$ is indeed winning.
\end{proof}

Second, observe that, given an ACA $\mathcal{A}$, two strategies $s$ and $t$ for a process $p_i$, and an input sequence $\gamma \in (2^\inputs{i})^\omega$, every strategy for \dom in the \iddominance game $(\mathcal{A},\computation{t}{\gamma},\computation{s}{\gamma})$ game corresponds to a run tree of $\mathcal{A}$ on $\computation{s}{\gamma}$.
To formalize this, we first define a \emph{projected play} of the \iddominance game.

\begin{definition}[Projected Play]
	Let $\mathcal{A} = (Q,q_0,\delta,F)$ be an ACA. Let $s$ an $t$ be strategies for process $p_i$. Let $\gamma \in (2^\inputs{i})^\omega$. Let $\rho$ be some play in the \iddominance game $(\mathcal{A},\computation{t}{\gamma},\computation{s}{\gamma})$. The \emph{projected play} $\hat{\rho} \in (Q \times Q)^\omega$ is defined by $\hat{\rho}_j \coloneq \projection{\rho_{4j}}{1}$ for all $j \in \mathbb{N}$.
	The \emph{projected dominant play} $\hat{\rho}^d \in Q^\omega$ of $\rho$ and the \emph{projected alternative play} $\hat{\rho}^a \in Q^\omega$ of $\rho$ are defined by $\hat{\rho}^d_j \coloneq \projq{\rho_{4j}}$ and $\hat{\rho}^a_j \coloneq \projp{\rho_{4j}}$ for all $j \in \mathbb{N}$, respectively.
\end{definition}

Intuitively, we obtain $\hat{\rho}$ from $\rho$ by removing all positions that are not of the form $((p,q),j)$ and by projecting to the state tuple, thus removing the index $j$.
The projected dominant play $\hat{\rho}^d$ is then obtained by further projecting to the the dominant state of the state tuples pf $\hat{\rho}$, \ie, to $q$ for a state tuple $(p,q)$, while we further project to the alternative state in the projected alternative play $\hat{\rho^a}$, \ie, to $p$ for a state tuple $(p,q)$.
Now, we can formalize the correspondence between strategies in the \iddominance game and run trees:

\begin{lemma}\label{lem:strategy_induces_runtree_dom}
	Let $\mathcal{A} = (Q,q_0,\delta,F)$ be an ACA. Let $s$ an $t$ be strategies for some process~$p_i$. Let $\gamma \in (2^\inputs{i})^\omega$. Let $\tau$ be a strategy for \dom in the \iddominance game $(\mathcal{A},\computation{t}{\gamma},\computation{s}{\gamma})$. Let $\mathcal{P}_\tau$ be the set of initial plays that are consistent with $\tau$. Then, there exists a run tree $r$ of $\mathcal{A}$ induced by $\computation{s}{\gamma}$ such that we have $\mathcal{R}(r) = \{ \hat{\rho}^d \mid \rho \in \mathcal{P}_\tau \}$, where $\mathcal{R}(r)$ is the set of infinite branches of $r$.
\end{lemma}
\begin{proof}
	We construct a $Q$-labeled tree $(\mathcal{T},\ell)$ from $\tau$ as follows by defining the labeling of the root as well as of the successors of all nodes.
	The labeling of the root $\varepsilon$ of $\mathcal{T}$ is defined by $\ell(\varepsilon)=q_0$. For a node $x \in \mathcal{T}$ with depth $j = |x|$, we define the labeling of the successor nodes of $x$ such that $\{ \ell(x') \mid x' \in \children{x} \} = \{ \hat{\rho}^d_{j+1} \mid \rho \in \mathcal{P}_\tau \land \compatible{\rho}{x} \}$ holds, where $\compatible{\rho}{x}$ denotes that~$\rho$ and~$x$ are compatible in the sense that for all $0 \leq j' \leq j$, we have $\hat{\rho}^d_{j'} = \ell(a_{j'})$, where~$a$ is the unique finite sequence of nodes that, starting from $\varepsilon$, reaches $x$. 
	
	Next, we show that $(\mathcal{T},\ell)$ is a run tree of $\mathcal{A}$ induced by $\computation{s}{\gamma}$.
	For the sake of readability, let $\sigma^s := \computation{s}{\gamma}$. Let $x \in \mathcal{T}$ be some node.
	Then, by construction of the \iddominance game, we know that \dom controls the existential transitions of $\mathcal{A}$ for $\computation{s}{\gamma}$, while the universal ones are controlled by \alt. Hence, since $\tau$ is a strategy of \dom, $\tau$ defines the existential choices in $\mathcal{A}$ for $\computation{s}{\gamma}$. 
	Therefore, for every round of the \iddominance game and thus for every time step $j \in \mathbb{N}$, there exists a decision for the existential choices in $\mathcal{A}$ for $\computation{s}{\gamma}$, namely the one defined by $\tau$, such that all initial plays that are consistent with $\tau$ adhere to it. Moreover, as no strategy for \alt is given, for every round of the game the set of plays in $\mathcal{P}_\tau$ defines \emph{all} possible universal choices in $\mathcal{A}$ for $\computation{s}{\gamma}$ that fit in with the existential choice defined by $\tau$ as well as the history.
	Hence, we obtain that $\{ \hat{\rho}^d_{j+1} \mid \rho \in \mathcal{P}_\tau \land \compatible{\rho}{x} \}$ satisfies $\bigvee_{c' \in \delta(q,\sigma^s_j)} \bigwedge_{q' \in c'} q'$, where $q := \ell(x)$. 
	Thus, by definition of run trees, $(\mathcal{T},\ell)$ is indeed a run tree of $\mathcal{A}$ induced by $\computation{s}{\gamma}$.
	Intuitively, the dominant states of an initial play that is consistent with $\tau$ thus evolve according to a run of $\mathcal{A}$ induced by $\computation{s}{\gamma}$.
	Furthermore, by construction of $(\mathcal{T},\ell)$, we immediately obtain that $\mathcal{R}(r) = \{ \hat{\rho}^d \mid \rho \in \mathcal{P}_\tau \}$ holds, where $\mathcal{R}(r)$ is the set of infinite branches of $r$.
\end{proof}

Similarly, a strategy for \alt in the \iddominance game $(\mathcal{A},\computation{t}{\gamma},\computation{s}{\gamma})$ corresponds to a run tree in $\mathcal{A}$ induced by $\computation{t}{\gamma}$. The formal formulation as well as the proof of this observation are analogous to \Cref{lem:strategy_induces_runtree_dom}. Therefore, we omit it.

Vice versa, we can translate a run tree of an ACA $\mathcal{A}$ induced by $\computation{t}{\gamma}$ for some strategy $t$ for process $p_i$ and an input sequence $\gamma \in (2^\inputs{i})$ into a strategy for \alt in the \iddominance game $(\mathcal{A},\computation{t}{\gamma},\computation{s}{\gamma})$ for some strategy $s$ for $p_i$:

\begin{lemma}\label{lem:runtree_induces_strategy_alt}
	Let $\mathcal{A} = (Q,q_0,\delta,F)$ be an ACA. Let $s$ an $t$ be strategies for some process~$p_i$. Let $\gamma \in (2^\inputs{i})^\omega$. Let $r$ be a run tree of $\mathcal{A}$ induced by $\computation{t}{\gamma}$. Let $\mathcal{R}(r)$ be the set of infinite branches of $r$. Then, there exists a strategy~$\mu$ for \alt in the \iddominance game $(\mathcal{A},\computation{t}{\gamma},\computation{s}{\gamma})$ such that $\mathcal{R}(r) = \{ \hat{\rho}^a \mid \rho \in \mathcal{P}_\mu \}$, where $\mathcal{P}_\mu$ be the set of initial plays that are consistent with $\mu$.
\end{lemma}
\begin{proof}
	We construct a strategy $\mu$ for \alt in the \iddominance game from $r$ as follows.
	Let $\eta \cdot \delta$ be a finite sequence of positions of the game with $\eta \in V^*$ and $\delta \in V$. We only define~$\mu$ explicitly on sequences $\eta \cdot \delta$ that can occur in the \iddominance game $(\aca{\varphi},\computation{t}{\gamma},\computation{s}{\gamma})$ and where $\delta$ is controlled by \alt; on all other sequences we define $\mu(\eta\cdot\delta) = v$ for some arbitrary $v \in V$ that is a valid extension of $\eta \cdot \delta$. Thus, in the following we assume that~$\eta\cdot\delta$ is a prefix that can occur in the game and that $\delta$ is of the form $((p,q),j)$ or $((p,q,c,c'),j)$. We map~$\eta \cdot \delta$ to a prefix of a branch of $r$ if there is a compatible one: 	
	a compatible branch $b$ of $r$ agrees with the finite projected play $\hat{\eta}$ up to point in time~$|\eta|$. Note here that, slightly misusing notation, we apply the definition of a projected play also to the finite prefix $\eta$ of a play.
	Moreover, no matter whether $\delta$ is of the form $((p,q,c),j)$ or $((p,q,c,c'),j)$, we have  $b_{|\eta|+1} = (p,q)$.
	
	If there is no compatible branch in $r$, we again define $\mu(\eta\cdot\delta) = v$ for some arbitrary $v \in V$ that is a valid extension of $\eta \cdot \delta$.
	Otherwise, the successors of $p$ in $b$ define the choice of~$\mu$: by definition, the set $\mathcal{S}$ of successors of the node labeled with $p$ in $b$ satisfies $\delta((p),\sigma^t_{|\eta|+1})$, where $\sigma^t := \computation{t}{\gamma}$. 
	Thus, there exists some $c \in \delta(p,\sigma^t_{|\eta|+1})$ such that for all $p' \in c$ we have $p' \in \mathcal{S}$.
	If $\delta = ((p,q),j)$, we thus define $\mu(\eta\cdot\delta) = ((p,q,c),j)$.
	If $\delta = ((p,q,c',c),j)$, we define $\mu(\eta\cdot\delta) = v$ for some arbitrary $v \in V$ that is a valid extension of $\eta \cdot \delta$. Note here that choosing an arbitrary successor for $\eta\cdot\delta$ for $\mu$ is possible since the choice defines a successor state for the dominant state $q$. Hence, the choice does not influence the projected alternative play.
	Since the existential choices in $\mathcal{A}$ define the run tree, it immediately follows from the construction of $\mu$ that $\mathcal{R}(r) = \{ \hat{\rho}^a \mid \rho \in \mathcal{P}_\mu \}$ holds. 
\end{proof}

Similarly, a run tree of $\mathcal{A}$ induced by $\computation{s}{\gamma}$ corresponds to a strategy for \dom in the \iddominance game $(\mathcal{A},\computation{t}{\gamma},\computation{s}{\gamma})$. The formal formulation as well as the proof of this observation are analogous to \Cref{lem:runtree_induces_strategy_alt}. Therefore, we omit it.

\subsection{\texorpdfstring{\iDDominance}{Delay-Dominance} implies Dominance (Proof of \texorpdfstring{\Cref{lem:ddom_implies_dom}}{Theorem 5})}

With the observations introduced above, we are now able to prove \Cref{lem:ddom_implies_dom}, \ie, that every \iddominant strategy is remorsefree dominant as well.

\begin{proof}
	Let $\aca{\varphi} = (Q,q_0,\delta,F)$.
	Suppose that $s$ is \iddominant for $\aca{\varphi}$, while $s$ is not dominant for $\varphi$. Then, there exists an alternative strategy $t$ for process $p_i$ and an input sequence $\gamma \in (2^\inputs{i})^\omega$ such that $\computation{s}{\gamma} \not\models \varphi$ holds, while we have $\computation{t}{\gamma} \models \varphi$. 
	Since $s$ is \iddominant for~$\aca{\varphi}$ by assumption, there is a winning strategy~$\tau$ for \dom in the \iddominance game $(\aca{\varphi},\computation{t}{\gamma},\computation{s}{\gamma})$.
	
	First, since $\computation{s}{\gamma} \not\models \varphi$ holds by assumption, all run trees of~$\aca{\varphi}$ induced by $\computation{s}{\gamma}$ contain a branch that visits infinitely many rejecting states. 
	By \Cref{lem:strategy_induces_runtree_dom}, there is a run tree $r^s$ of~$\aca{\varphi}$ induced by $\computation{s}{\gamma}$ that reflects the choices for the existential transitions of~$\aca{\varphi}$ for $\computation{s}{\gamma}$ defined by $\tau$. Moreover, we have $\mathcal{R}(r^s) = \{ \hat{\rho}^d \mid \rho \in \mathcal{P}_\tau \}$, where $\mathcal{R}(r^s)$ is the set of infinite branches of $r^s$ and where $\mathcal{P}_\tau$ is the set of initial plays that are consistent with~$\tau$.
	Thus, since $r^s$ contains a branch that visits infinitely many rejecting states by assumption, there is an initial play $\rho$ that is consistent with $\tau$ such that $\hat{\rho}^d$ visits infinitely many rejecting states. Therefore, by definition of $\hat{\rho}^d$, $\rho$ contains infinitely many rejecting dominant states.
	
	Next, since $\computation{t}{\gamma} \models \varphi$ holds by assumption, there is a run tree of $\aca{\varphi}$ induced by $\computation{t}{\gamma}$ whose branches all visit only finitely many rejecting states. Let $r^t$ be this run tree. 
	By \Cref{lem:runtree_induces_strategy_alt}, there is a strategy $\mu$ for \alt in the \iddominance game that reflects the choices of $r^t$ for the existential transitions of~$\aca{\varphi}$ for $\computation{t}{\gamma}$. Moreover, we have $\mathcal{R}(r^t) = \{ \hat{\rho}^a \mid \rho \in \mathcal{P}_\mu \}$, where $\mathcal{R}(r^t)$ is the set of infinite branches of $r^t$ and where $\mathcal{P}_\mu$ is the set of initial plays that are consistent with~$\mu$.	
	Thus, since all branches of $r^t$ visit only finitely many rejecting states, it immediately follows that for all initial plays that are consistent with~$\mu$ we have that $\hat{\rho}^a$ visits only finitely many rejecting states.
	
	Therefore, in particular for the initial play that is consistent with both $\tau$ and $\mu$, it holds that $\hat{\rho}^a$ visits only finitely many rejecting states.
	However, as shown above, $\hat{\rho}^d$ visits infinitely many rejecting states. Thus, there is a point in time $j \in \mathbb{N}$ such that $\projq{\rho_j} \in F$, while $\projp{\rho_{j'}} \not \in F$ for all $j' \geq j$.
	But then $\tau$ is not a winning strategy for \dom; yielding a contradiction and thus proving the claim.
\end{proof}


\subsection{Bad Prefixes for \texorpdfstring{\iDDominance}{Delay-Dominance}}

A critical shortcoming of remorsefree dominance is its non-compositionality for liveness properties. This restricts the usage of dominance-based compositional synthesis algorithms to safety specifications, which are in many cases not expressive enough.
\iDdominance, in contrast, is specifically designed to be compositional for more properties. 
This heavily relies on two facts:~(i)~\iddominance is not defined using the satisfaction of the given specification but on a more involved property on the visits of rejecting states, and~(ii)~\iddominance is defined using a two-player game and thus we require the existence of a \emph{strategy} for \dom, \ie, determining  which decisions to make for the existential choices of the \iddominant strategy and the universal ones for the alternative strategy has to be possible without knowledge about the future input as well as the future decisions for the other choices. {}
More precisely, compositionality requires that whenever the parallel composition of two strategies $s_1$ and $s_2$ for processes $p_1$ and $p_2$, respectively, does not satisfy the strategy property~--~that is, for instance, remorsefree dominance or \iddominance~--~we are able to \emph{blame} at least one of the processes for being responsible for violating that strategy property. Otherwise, none of the processes ever behaves incorrectly with respect to the strategy property and thus none of the processes violates the strategy property. Hence, then the parallel composition of two strategies that satisfy the strategy property would not necessarily satisfy it as well.

As an example reconsider the message-sending system from the running example and the strategy property remorsefree dominant. Furthermore, consider the strategies~$s_1$ and~$s_2$ that wait for receiving the respective other message before sending their own one. The parallel composition $s_1 \pc s_2$ never sends any message and thus violates the specification $\varphi = \Eventually m_1 \land \Eventually m_2$ on every input sequence. Yet, none of the processes can be blamed for being responsible for violating the properties of remorsefree dominance: 
even if process $p_i$ would send its message $m_i$ eventually, the specification is still not satisfied since message~$m_{3-i}$ has not been send yet. Thus, as long as it did not receive $m_{3-i}$, it is not required to eventually send~$m_i$. The same, however, also holds for system process $p_{3-i}$.
Note here that it is crucial that although~$p_i$ is required to send $m_i$ eventually if it receives~$m_{3-i}$, process $p_{3-i}$ is not required to send~$m_{3-i}$ in the first place; resulting in the deadlock situation where both processes wait on each other indefinitely. Nevertheless, both $s_1$ and $s_2$ are both remorsefree dominant since they are not required to output their message when confronted with the behavior defined by the other process in a computation of $s_1 \pc s_2$.

Intuitively, we can blame at least one of the processes $p_1$ and $p_2$ for violating the strategy property if there exists a bad prefix of a computation of the parallel composition $s_1 \pc s_2$ of the process strategies~$s_1$ and~$s_2$ for the strategy property, \ie, a prefix of a computation of $s_1 \pc s_2$ such that all of its infinite extensions violate the strategy property.
For remorsefree dominance, for instance, a bad prefix of $\computation{s_1 \pc s_2}{\gamma}$ is a finite prefix $\eta$ of $\computation{s_1 \pc s_2}{\gamma}$ such that all infinite extensions $\sigma$ of $\eta$ that respect $\gamma$, \ie, that agree with $\gamma$ on the inputs of $p_1 \pc p_2$, violate the specification while there exists an alternative strategy $t$ for $p_1 \pc p_2$ such that $\computation{t}{\gamma}$ satisfies it. Note here that since remorsefree dominance only considers the satisfaction of a specification, the existence of a bad prefix for remorsefree dominance boils down to the existence of a bad prefix for the considered specification. Clearly, there do not exist bad prefixes for remorsefree dominance for liveness properties.

\iDdominance, in contrast, takes an alternating co-Büchi automaton describing the specification into account and relates the visit to a rejecting states induced by the possibly dominant strategy to those induced by an alternative strategy. Thus, non-existence of bad prefixes for liveness properties does not necessarily result in the absence of bad prefix for \iddominance.
For instance, reconsider the message sending system and the ACA $\mathcal{A}_\varphi$ depicted in \Cref{fig:ACA_running_example}, which describes the specification $\varphi$. Although $\Lang{\varphi}$ is a liveness property and thus does not have a bad prefix, $\mathcal{A}_\varphi$ ensures bad prefixes for \iddominance:
let $s_i$ be a strategy for $p_i$ that is not \iddominant for $\mathcal{A}_\varphi$. Then, there exists an input sequence $\gamma \in (2^\inputs{i})^\omega$ and an alternative strategy $t_i$ for $p_i$ such that \dom does not have a winning strategy in the \iddominance game $(\mathcal{A}_\varphi,\computation{t_i}{\gamma},\computation{s_i}{\gamma})$. 
Note that, as outlined in \Cref{sec:ddominance}, for such an input sequence $\gamma \in (2^\inputs{i})^\omega$ it holds that $m_{3-i}$ occurs in $\computation{s_i}{\gamma}$ before~$m_i$. The prefix of $\computation{s_i}{\gamma}$ up to the point in time $j \in \mathbb{N}$ at which~$m_{3-i}$ occurs while $m_i$ did not occur so far is then a bad prefix of $\computation{s_i}{\gamma}$ in the sense of \Cref{def:bad_prefixes}: since $\mathcal{A}_\varphi$ is deterministic, every play in the \iddominance game stays in state $q_0$ in the dominant states up to the point in time $j$ at which $m_{3-i}$ occurs and then moves to either $q_1$ or $q_2$, depending on whether $i=1$ or $i=2$ holds. It then stays there until~$m_i$ occurs and moves to $q_3$ afterwards. An alternative strategy that outputs $m_i$ in $\gamma$ at point in time~$j$, \ie, at the very same point in time at which $m_{3-i}$ occurs, moves from $q_1$ directly to $q_3$, thus omitting the visit of a rejecting state $q_1$ or $q_2$. Hence, no matter how~$p_1$ behaves after entering $q_1$ or $q_2$, respectively, there is an alternative strategy that causes \dom to lose the \iddominance game for input sequence $\gamma$.

\begin{figure}[t]
	\centering
	\scalebox{0.96}{
	\begin{tikzpicture}[>=latex,shorten >=0pt,auto,->,node distance=1cm,semithick,every edge/.style={draw,font=\normalsize}, initial text = ]

		\def\centerarc[#1](#2)(#3:#4:#5)
			 { \draw[#1] ($(#2)+({#5*cos(#3)},{#5*sin(#3)})$) arc (#3:#4:#5); }
	
		\node[state,initial]		(q0)		at (0,0)		{$q_0$}; 
		\node[state,accepting]	(q1)		at (2.5,0.5)	{$q_1$};
		\node[state]			(q5)		at (6.75,-1)	{$q_5$};
		\node[state]			(q2)		at (5,0.5)	{$q_2$}; 
		\node[state,accepting]	(q3)		at (8.5,0.5)	{$q_3$}; 
		\node[state]			(q4)		at (11,0.5)	{$q_4$}; 
		
		\node[draw=none]		(a1)		at ($(q2.east)+(1,0.13)$) {};
		\node[draw=none]		(a2)		at ($(q5.east)+(0.19,0.79)$) {};
		\centerarc[gray!80,-,thick](a1.south)(0:-117:0.17)
		\centerarc[gray!80,-,thick](a2.south east)(40:-76:0.19)
		
		\path	(q0)		edge	[sloped]			node		{$a$}			(q1)
						edge	[sloped,below]	node		{$\neg a$}		(q5)
				(q1)		edge					node		{$\top$}			(q2)
				(q2)		edge					node		{$\top$}			(q3)
						edge[bend left]		node		{$a$}			(q4)
				(q3)		edge					node		{$\top$}			(q4)
				(q4)		edge[loop right]		node		{$\top$}			(q4)
				(q5)		edge[sloped]	node		{$\top$}			(q3);
				
		\path	(a1.south)		edge[bend left=50]	node		{}	(q2)
				(a2.south east)	edge[bend left=50]	node		{}	(q5);
	\end{tikzpicture}}
	\caption{Alternating co-Büchi automaton $\mathcal{A}$ over alphabet $\{a,b\}$ that does not ensure bad prefixes for \iddominance. Universal choices are depicted with a gray arc.}\label{fig:ACA_no_bad_prefixes}
\end{figure}
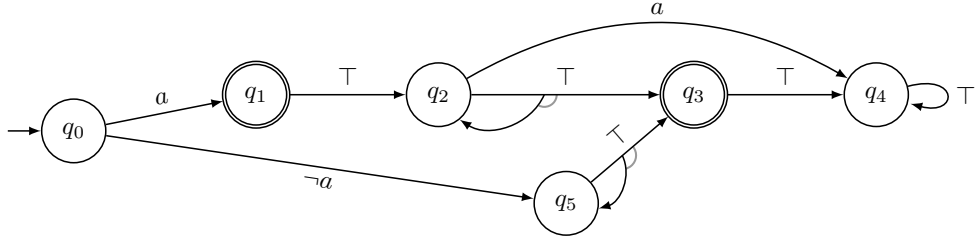
As already pointed out above, there are much more properties for which there exists an alternating co-Büchi automaton that ensures bad prefixes for \iddominance than properties that have a \enquote{classical} bad prefix: no liveness property has a classical bad prefix, yet, for many of them there exist ACAs that ensure bad prefixes for \iddominance.
The ACA $\mathcal{A}$ depicted in \Cref{fig:ACA_no_bad_prefixes} does not ensure bad prefixes, though: let $b$ be an input variable and let $a$ be an output variable.
let $s$ be a strategy that outputs $a$ in the very first time step and never outputs~$a$ afterwards. That is, irrespective of the input sequence $\gamma \in (2^{\{b\}})^\omega$, the computation of~$s$ is given by $\computation{s}{\gamma} = \{a\}\emptyset^\omega$.
Let $t$ be an alternative strategy that never outputs~$a$, \ie, $\computation{t}{\gamma} = \emptyset^\omega$ holds for all $\gamma \in (2^{\{b\}})^\omega$. 
For an arbitrary input sequence $\gamma$, consider the \iddominance game $\mathcal{G} = (\mathcal{A},\computation{t}{\gamma},\computation{s}{\gamma})$.
\dom does not have a winning strategy: in the \iddominance game, \dom controls the universal choice to  to either stay in state $q_5$ or tho move to $q_3$ in the alternative states. \alt, in contrast, controls the universal choice to either stay in $q_2$ or to move to $q_3$. Note that for the considered sequences $\computation{t_i}{\gamma}$ and $\computation{s_i}{\gamma}$ no existential choices occur.
As soon as \dom chooses to let the alternative states move from $q_5$ to $q_3$ while the dominant states are still in $q_1$ or $q_2$, \alt can choose to let the dominant states move from $q_2$ to $q_3$ sometime afterwards; resulting in a visit to rejecting dominant state, namely $q_3$, that is never matched by a rejecting alternative state.
If \dom chooses to let the alternative states stay in $q_5$, however, \alt has the possibility to let the dominant states stay in $q_2$ as well. Then, there is a visit to a rejecting dominant state, namely $q_1$, that is never answered if \dom never lets the alternative states move to $q_3$. As argued above, choosing to move to $q_3$ in the alternative states while the dominant states are still in $q_2$ results in a rejecting dominant state that is never matched as well, though.
Hence, neither letting the alternative states stay in $q_5$ forever nor letting them move to $q_3$ eventually results in a winning strategy for \dom in $\mathcal{G}$.
Thus, $s$ is not \iddominant.

Yet, there does not exist a bad prefix of \iddominance for $s_i$: let $k \geq 2$ be some point in time and let $\eta$ be the prefix of $\computation{s_i}{\gamma}$ up to point in time $k$, \ie, let $\eta \coloneq \pref{\computation{s_i}{\gamma}}{k+1}$. Let $\sigma$ be an infinite extension of $\eta$ with $a \in \sigma_{k'}$ for some point in time $k' \geq k$ and consider the \iddominance game $\mathcal{G}' = (\mathcal{A},\computation{t}{\gamma},\sigma)$.
	We construct a winning strategy $\tau$ for \dom in $\mathcal{G}'$ as follows: for the existential choice in $q_2$ at point in time $k'$, \ie, at the point in time at which $a$ occurs, $\tau$ chooses to let the dominant states move to $q_4$. For the universal choice in $q_5$, it chooses to let the alternative states stay in $q_5$ up to point in time $k'-1$ and to let them move to $q_3$ afterwards, \ie, at point in time $k'$.
	Then, $\tau$ ensures the visit to a rejecting alternative state \emph{after} point in time $k'$, namely at point in time $k'+1$. The last rejecting dominant state, however, occurs \emph{before} point in time $k'$:
	for an initial play $\rho$ that is consistent with $\tau$ and in which the dominant states are in $q_2$ at point in time $k'$, \ie, with $\projq{\rho_{k'}} = q_2$, \dom's strategy $\tau$ ensures that no rejecting dominant state is visited after point in time~$k'$. In fact, no rejecting dominant state is visited after point in time $1$, at which the dominant states visited $q_1$, since, by construction of $\mathcal{A}$, state $q_2$ is non-rejecting, it is reached at point in time $2$ and staying in $q_2$ is the only possibility to be in $q_2$ at point in time $k'$. Since $k \geq 2$ and $k' \geq k$ holds by construction, we clearly have $1 < k'$ and thus the last rejecting dominant state occurs before point in time $k'$. 
	For an initial play $\rho$ that is consistent with $\tau$ and in which the dominant states move from $q_2$ to $q_3$ at some point in time $k'' < k'$, \ie, with $\projq{\rho_{k''}} = q_2$ and $\projq{\rho_{k''+1}} = q_3$, it follows from the construction of $\mathcal{A}$ that the last rejecting dominant state is visited at point in time $k''+1$. Hence, since $k'' < k'$ holds by construction, the last rejecting dominant state occurs before point in time $k'$.
	Thus, in every initial play $\rho$ that is consistent with $\tau$, every visit to a rejecting dominant state is matched by a visit to a rejecting dominant state and thus $\rho \in \win$ holds. Hence, $\tau$ is a winning strategy for \dom in $\mathcal{G}'$.
	Since we chose $k \geq 2$ arbitrarily, there thus does not exist a bad prefix for \iddominance in $\mathcal{A}$.


\subsection{Compositionality of \texorpdfstring{\iDDominance}{Delay-Dominance} (Proof of \texorpdfstring{\Cref{thm:compositonality_ddom}}{Theorem 7})}

The existence of a bad prefix for \iddominance allows us to prove \Cref{thm:compositonality_ddom}, \ie, that the parallel composition of two \iddominant strategies is \iddominant as well if the ACA ensure bad prefixes for \iddominance.
Similar to the proof of the compositionality of remorsefree dominance for safety specifications, the proof is by contradiction.
That is, we suppose that $s_1 \pc s_2$ is not \iddominant for~$\mathcal{A}$ and $p_1 \pc p_2$. Then, there is an alternative strategy~$t$ for $p_1 \pc p_2$ and an input sequence $\gamma \in (2^\inputs{1,2})^\omega$ such that there is no winning strategy for \dom in the \iddominance game $(\mathcal{A},\computation{s_1 \pc s_2}{\gamma},\computation{t}{\gamma})$.
The existence of a bad prefix for \iddominance then allows us to blame one of the processes for preventing \iddominance. Thus, we make a case distinction on whether $p_1$ is solely responsible for \dom losing the \iddominance game or whether (at least also) $p_2$ is responsible. With the properties of process strategies as well as the fact that they are represented by Moore machines, we can conclude in both cases that the strategy of the respective process cannot be \iddominant.

\begin{proof}
	For the sake of readability, let $\inputs{1,2} \coloneq (\inputs{1} \cup \inputs{2})\setminus(\outputs{1} \cup \outputs{2})$ be the set of inputs of $p_1 \pc p_2$, let $\outputs{1,2} \coloneq \outputs{1} \cup \outputs{2}$ be the set of outputs of $p_1 \pc p_2$, and let $\variables{1,2} \coloneq \inputs{1,2} \cup \outputs{1,2}$ be the set of $p_1 \pc p_2$'s variables.
	Let $\mathcal{A} = (Q,q_0,\delta,F)$.	
	Suppose that~$s_1 \pc s_2$ is not \iddominant for~$\mathcal{A}$ and $p_1 \pc p_2$. Then, there is an alternative strategy~$t$ for $p_1 \pc p_2$ and an input sequence $\gamma \in (2^\inputs{1,2})^\omega$ such that there is no winning strategy for \dom in the \iddominance game $\mathcal{G} = (\mathcal{A},\computation{s_1 \pc s_2}{\gamma},\computation{t}{\gamma})$. {}
	
	\enlargethispage{\baselineskip}
	By assumption, $\mathcal{A}$ ensures bad prefixes for \iddominance and thus, in particular, there exists a finite prefix $\nu \in (2^{\variables{1,2}})^*$ of $\computation{s_1\pc s_2}{\gamma}$ such that for all infinite extensions~$\sigma$ of $\nu$, there exists some infinite sequence $\sigma' \in (2^{\variables{1,2}})^\omega$ such that \dom loses the \iddominance game $(\mathcal{A},\sigma',\sigma)$.
	Thus, in particular, there exists a \emph{smallest} such prefix. Let $\eta \cdot \delta \in (2^{\variables{1,2}})^*$ be this smallest, \ie, shortest, such prefix, where $\eta \in (2^{\variables{1,2}})^*$ and $\delta \in 2^{\variables{1,2}}$. Since $\eta \cdot \delta$ is a bad prefix for \iddominance, it holds that for all infinite extensions $\sigma$ of $\eta \cdot \delta$, there exists some infinite sequence $\sigma' \in (2^{\variables{1,2}})^\omega$ with $\sigma' \cap \inputs{1,2} = \sigma \cap \inputs{1,2}$ such that \dom loses the \iddominance game $(\mathcal{A},\sigma',\sigma)$.
	Furthermore, since $\eta \cdot \delta$ is the smallest such prefix, there exists an infinite extension $\hat{\sigma} \in (2^{\variables{1,2}})^\omega$ of $\eta$ such that \dom wins the \iddominance game $(\mathcal{A},\sigma',\hat{\sigma})$ for all $\sigma' \in (2^{\variables{1,2}})^\omega$ with $\sigma' \cap \inputs{1,2} = \hat{\sigma} \cap \inputs{1,2}$.
	Note that $\eta \cdot \delta$ cannot be the empty sequence as otherwise, for all infinite sequences $\sigma \in (2^{\variables{1,2}})^\omega$ that agree on $p_1 \pc p_2$'s input with $\gamma$, \dom does not have a winning strategy in the \iddominance game $(\mathcal{A},\computation{t}{\gamma},\sigma)$. But then \dom particularly does not have a winning strategy in the \iddominance game $(\mathcal{A},\computation{t}{\gamma},\computation{t}{\gamma})$; contradicting that every strategy \iddominates itself (see \Cref{lem:ddom_self_dominance}).
	Let $m := |\eta \cdot \delta|$ be the length of $\eta \cdot \delta$.
	The last position $\delta$ of the prefix $\eta \cdot \delta$ contains decisions of both processes $p_1$ and $p_2$ defined by their strategies $s_1$ and $s_2$. We distinguish the following two cases:
	\begin{enumerate}
		\item There is an infinite extension $\sigma \in (2^\variables{1,2})^\omega$ of $\eta$ with $\sigma_{m-1} \cap (\variables{1,2} \setminus \outputs{1}) = \delta \cap (\variables{1,2}\setminus\outputs{1})$ and $\sigma \cap \inputs{1,2} = \gamma$ such that \dom has a winning strategy in the \iddominance game $(\mathcal{A},\computation{t}{\gamma},\sigma)$. Hence, intuitively, it is the fault of process $p_1$ and thus, in particular, of its strategy $s_1$, that \dom loses the game $\mathcal{G}$. Let~$t'$ be a strategy for $p_1 \pc p_2$ that produces~$\sigma$ on input sequence~$\gamma$, \ie, a strategy with $\computation{t'}{\gamma} = \sigma$. Furthermore, let $\gamma^{s_2} = \computation{s_1 \pc s_2}{\gamma} \cap \outputs{2}$ and let $\gamma^{t'_2} = \computation{t'}{\gamma} \cap \outputs{2}$. Since we have $\delta \cap ({\variables{1,2}} \setminus \outputs{1}) = \sigma_{m-1} \cap ({\variables{1,2}} \setminus \outputs{1})$ by assumption, we obtain that $\eta \cdot \delta$ and $\pref{\sigma}{m}$ agree on the variables in $\variables{1,2} \setminus \outputs{1}$ and thus, in particular, $\pref{\computation{s_1\pc s_2}{\gamma}}{m}$ and $\pref{\computation{t'}{\gamma}}{m}$ agree on the variables in $\variables{1,2} \setminus \outputs{1}$. Hence, it follows with the construction of $\gamma^{s_2}$ and $\gamma^{t'_2}$ that $\pref{(\gamma \cup \gamma^{s_2})}{m} = \pref{(\gamma \cup \gamma^{t'_2})}{m}$ holds. Since strategies cannot look into the future,~$s_1$ thus cannot behave differently on input sequences $\gamma \cup \gamma^{s_2}$ and $\gamma \cup \gamma^{t'_2}$ up to point in time $m-1$. Hence, $\pref{\computation{s_1}{\gamma \cup \gamma^{s_2}}}{m-1} = \pref{\computation{s_1}{\gamma \cup \gamma^{t'_2}}}{m-1}$ follows. Since $\eta\cdot\delta$ is a finite prefix of $\computation{s_1 \pc s_2}{\gamma}$ and since we have $\computation{s_1}{\gamma\cup\gamma^{s_2}} = \computation{s_1 \pc s_2}{\gamma}$ by construction of~$\gamma^{s_2}$ and by definition of computations, $\computation{s_1}{\gamma\cup\gamma^{t'_2}}$  is an infinite extension of $\eta\cdot\delta$. Furthermore, since the variables in $\inputs{1,2}$ are solely defined by~$\gamma$, it follows immediately from the definition of computations that it agrees with~$\gamma$ on these variables. By construction of $\eta \cdot \delta$, \dom thus loses the \iddominance game $\mathcal{G}' = (\mathcal{A},\computation{t}{\gamma},\computation{s_1}{\rho^{t'}\cap\inputs{1}})$. Yet, by construction of~$t'$, \dom has a winning strategy $\tau$ in the game $(\mathcal{A},\computation{t}{\gamma},\computation{t'}{\gamma})$. Let~$t'_1$ be a strategy for $p_1$ such that $\computation{t'}{\gamma} = \computation{t'_1}{\gamma\cup\gamma^{t_2}}$ holds. Then, since~$s_1$ is \iddominant for $p_1$ and $\mathcal{A}$ by assumption, it particularly \iddominates~$t'_1$ on input $\gamma \cup \gamma^{t'_2}$ and therefore \dom has a winning strategy~$\tau'$ in the \iddominance game $(\mathcal{A},\computation{t_1}{\gamma\cup\gamma^{t'_2}},\computation{s_1}{\gamma\cup\gamma^{t'_2}})$. Since $\computation{t'_1}{\gamma\cup\gamma^{t'_2}} = \computation{t}{\gamma}$ holds by construction of $t'_1$, we can thus combine~$\tau$ and~$\tau'$ to a strategy~$\tau''$ for \dom in the \iddominance game~$\mathcal{G}'$. Furthermore, since $\tau$ and $\tau'$ are winning in the respective games, it follows that for all initial plays $\rho$ that are consistent with~$\tau''$ it holds that whenever $\projq{\rho_j}\in F$ holds for a point in time $j \in \mathbb{N}$, then there is a point in time $j' \geq j$ such that $\projp{\rho_{j'}} \in F$ holds. Thus, $\tau''$ is a winning strategy for \dom in the game $\mathcal{G}'$; contradicting that \dom loses $\mathcal{G}'$.
		\item There is no infinite extension $\sigma \in (2^\variables{1,2})^\omega$ of $\eta$ with $\sigma_{m-1} \cap (\variables{1,2} \setminus \outputs{1}) = \delta \cap (\variables{1,2}\setminus\outputs{1})$ and $\sigma \cap \inputs{1,2} = \gamma$ such that \dom has a winning strategy in the \iddominance game $(\mathcal{A},\computation{t}{\gamma},\sigma)$. Hence, intuitively, it is (at least also) the fault of process $p_2$ and thus, in particular, of its strategy $s_2$, that \dom loses the game $\mathcal{G}$. By construction of $\eta \cdot \delta$, there exists an infinite extension $\sigma'$ of $\eta$ such that \dom has a winning strategy in the \iddominance game $(\mathcal{A},\computation{t}{\gamma},\sigma')$. Let~$t'$ be a strategy for $p_1 \pc p_2$ that produces~$\sigma'$ on input~$\gamma$, \ie, a strategy with $\computation{t'}{\gamma} = \sigma'$. Furthermore, let $\gamma^{s_1} = \computation{s_1 \pc s_2}{\gamma} \cap \outputs{1}$ and let $\gamma^{t'_1} = \computation{t'}{\gamma} \cap \outputs{1}$. Since $\computation{t'}{\gamma}$ is an infinite extension of~$\eta$ by definition of $t'$ and since $\eta \cdot \delta$ is a prefix of $\computation{s_1 \pc s_2}{\gamma}$ by construction of $\eta \cdot \delta$, we have $\pref{\computation{t'}{\gamma}}{m-1} = \pref{\computation{s_1 \pc s_2}{\gamma}}{m-1}$. Hence, it follows with the construction of $\gamma^{s_1}$ and $\gamma^{t'_1}$ that $\pref{(\gamma \cup \gamma^{s_1})}{m-1} = \pref{(\gamma \cup \gamma^{t'_1})}{m-1}$ holds. Since strategies cannot look into the future,~$s_2$ thus cannot behave differently on $\gamma \cup \gamma^{s_1}$ and $\gamma \cup \gamma^{t'_1}$ up to point in time $m-2$ and thus we have $\pref{\computation{s_2}{\gamma\cup\gamma^{s_1}}}{m-2} = \pref{\computation{s_2}{\gamma\cup\gamma^{t'_1}}}{m-2}$. Hence, $\computation{s_2}{\gamma\cup\gamma^{t'_1}}$ is an infinite extension of $\eta$. Since we consider process strategies that are represented by Moore machines, $s_2$ cannot react directly to an input. Thus, $\pref{\computation{s_2}{\gamma\cup\gamma^{s_1}}}{m-1} \cap \outputs{2} = \pref{\computation{s_2}{\gamma\cup\gamma^{t'_1}}}{m-1} \cap \outputs{2}$ holds. Furthermore, we have $\computation{s_2}{\gamma\cup\gamma^{s_1}} \cap \inputs{1,2} = \computation{s_2}{\gamma\cap\gamma^{t'_1}} \cap\inputs{1,2}$ by the definition of computations. By construction of $\inputs{1,2}$, $\outputs{1,2}$, and $\variables{1,2}$ as well as by definition of architectures, $\inputs{1,2} \cup \outputs{2} = \variables{1,2} \setminus \outputs{1}$ holds. Thus, $\pref{\computation{s_2}{\gamma\cup\gamma^{s_1}}}{m-1} \cap ({\variables{1,2}} \setminus \outputs{1}) = \pref{\computation{s_2}{\gamma\cup\gamma^{t'_1}}}{m-1} \cap ({\variables{1,2}} \setminus \outputs{1})$ follows. Let $\delta' = \computation{s_1}{\gamma\cup\gamma^{t'_2}}_{m-1}$. Then, $\delta \cap ({\variables{1,2}}\setminus\outputs{1}) = \delta' \cap ({\variables{1,2}}\setminus\outputs{1})$ holds since we have $\delta = \computation{s_1}{\gamma\cup\gamma^{s_2}}_{m-1}$ by construction of the prefix $\eta \cdot \delta$. Moreover, $\computation{s_1}{\gamma\cup\gamma^{t'_2}}$ is an infinite extension of $\eta \cdot \delta'$. By construction of $t'$, \dom has a winning strategy $\tau$ in the \iddominance game $(\mathcal{A},\computation{t}{\gamma},\computation{t'}{\gamma})$. Let $t'_2$ be a strategy for $p_2$ such that $\computation{t'}{\gamma} = \computation{t'_2}{\gamma \cup \gamma^{t'_1}}$ holds. Then, since $s_2$ is \iddominant for $p_2$ and $\mathcal{A}$ by assumption, it particularly \iddominates $t'_2$ on input $\gamma \cup \gamma^{t'_1}$ and therefore \dom has a winning strategy $\tau'$ in the \iddominance game $(\aca{\varphi},\computation{t'_2}{\gamma\cup\gamma^{t'_1}},\computation{s_2}{\gamma\cup\gamma^{t'_1}})$. Similar to the previous case, we can combine $\tau$ and $\tau'$ to a winning strategy $\tau''$ for \dom in the \iddominance game $(\aca{\varphi},\computation{t}{\gamma},\computation{s_2}{\gamma\cup\gamma^{t'_1}})$. Thus, $\computation{s_2}{\gamma\cup\gamma^{t'_1}}$ is an infinite extension $\sigma$ of $\eta$ with $\sigma_{m-1} \cap ({\variables{1,2}}\setminus\outputs{1}) = \delta \cap ({\variables{1,2}}\setminus\outputs{1})$ and $\sigma \cap \inputs{1,2} = \gamma$ such that \dom wins the \iddominance game $(\mathcal{A},\computation{t}{\gamma},\sigma)$; contradicting the assumption that no such infinite extension exists.
	\end{enumerate}
	Thus, no matter whether $p_1$ or $p_2$ is responsible for \dom losing the \iddominance game $(\mathcal{A},\computation{s_1 \pc s_2}{\gamma},\computation{t}{\gamma})$, the respective strategy $s_1$ or $s_2$ cannot be \iddominant as we obtain a contradiction. Hence, the claim that $s_1 \pc s_2$ is \iddominant for~$\mathcal{A}$ and $p_1 \pc p_2$ if $s_1$ and $s_2$ are \iddominant for~$\mathcal{A}$ and $p_1$ and $p_2$, respectively, follows.
\end{proof}


\section{Automaton Construction for \texorpdfstring{\iDDominant}{Delay-Dominant} Strategies}\label{app:automaton_construction}

In this section, we provide more details of the three-step automaton construction for synthesizing \iddominant strategies. Moreover, we give the proofs of the lemmas and theorems in \Cref{sec:automaton_construction} that we omitted in the paper due to space restrictions.
The construction of the alternating co-Büchi automaton $\ddACA{\varphi}$ relies heavily on the following observation:

\begin{lemma}\label{lem:disjunctive_ddominance}
	Let $\varphi$ be an LTL formula. Let $\mathcal{A}_{\varphi}$ be an ACA with $\Lang{\varphi} = \Lang{\mathcal{A}_\varphi}$. Let $s$ and~$t$ be strategies for process $p_i$. Let $\gamma \in (2^\inputs{i})^\omega\!$. Then, $\ddominatesSequence{s}{t}{\gamma}{\mathcal{A}_{\varphi}}$ for $\mathcal{A}_\varphi$ if, and only if, either~(i)~$\computation{t}{\gamma}\not\models\varphi$ holds or (ii) $\ddominatesSequence{s}{t}{\gamma}{\mathcal{A}_{\varphi}}$ and for the winning strategy $\tau$ of \dom in the \iddominance game we have for every initial play $\rho$ that is consistent with $\tau$ that there is a point in time $k$ such that $\projq{\rho_{k'}} \not\in F$ for all $k' \geq k$.
\end{lemma}
\begin{proof}
	First, let $\ddominatesSequence{s}{t}{\gamma}{\mathcal{A}_{\varphi}}$ hold for $\mathcal{A}_\varphi$. Then, there exists a winning strategy $\tau$ for \dom in the \iddominance game $(\mathcal{A}_\varphi,\computation{t}{\gamma},\computation{s}{\gamma})$.
	If, for every initial play $\rho$ that is consistent with $\tau$, $\rho$ has only finitely many visits to rejecting dominant states for every consistent play, then (ii) holds and thus the claim follows. 
	Otherwise, we have infinitely many visits to rejecting dominant states for some initial play $\rho$ that is consistent with $\tau$. Let $\mu$ be the strategy for \alt such that~$\rho$ is consistent with both~$\tau$ and~$\mu$.
	Note that, by construction of the \iddominance game, only the part of $\mu$ that defines the universal choices in $\mathcal{A}_\varphi$ for $\computation{s}{\gamma}$ affects whether or not $\rho$ contains infinitely many visits to rejecting dominant states. Let $\mu'$ be a strategy for only these choices that coincides with the ones defined by $\mu$.
	Then, for all full strategies $\mu''$ for \alt that coincide with $\mu'$ on the universal choices in $\mathcal{A}_\varphi$ for $\computation{s}{\gamma}$, the initial play $\rho'$ that is consistent with both~$\tau$ and~$\mu''$ contains infinitely many visits to rejecting dominant states.
	Since $\tau$ is a winning strategy for \dom by assumption and by construction of the \iddominance game, it follows that  every such initial play $\rho'$ contains infinitely many rejecting alternative states. Thus, intuitively, independent of the existential choices in $\mathcal{A}_\varphi$ for $\computation{t}{\gamma}$, $\tau$ can enforce infinitely many rejecting alternative states.
	
	By \Cref{lem:strategy_induces_runtree_dom}, or, more precisely, by the analogous lemma for strategies for \alt, for all such strategies $\mu''$, there exists a run tree $r_{\mu''}$ of $\mathcal{A}_\varphi$ induced by $\computation{t}{\gamma}$ that reflects the existential choices of $\mathcal{A}_\varphi$ for $\computation{t}{\gamma}$ defined by $\mu''$. Moreover, we have $\mathcal{R}(r) = \{ \hat{\rho}^a \mid \rho \in \mathcal{P}_{\mu''} \}$, where $\mathcal{R}(r)$ is the set of infinite branches of $r$ and where $\mathcal{P}_{\mu''}$ is the set of initial plays that are consistent with~$\mu$.
	Thus, by definition of the projected alternative play, we obtain that for all strategies $\mu''$ for \alt extending $\mu'$, the initial play $\rho'$ that is consistent with both $\tau$ and $\mu''$ is a branch of $r_{\mu''}$. Since $\rho'$ contains infinitely many rejecting alternative states, it follows that all such run trees $r_{\mu''}$ contain a branch with infinitely many visits to rejecting states.
	Moreover, since $\mu'$ does not fix any decision regarding the choices for $\computation{t}{\gamma}$, indeed \emph{every} run tree of $\mathcal{A}_\varphi$ induced by $\computation{t}{\gamma}$ contains a branch with infinitely many visits to rejecting states.
	Therefore, by definition of alternating co-Büchi automata, $\mathcal{A}_\varphi$ rejects $\computation{t}{\gamma}$. Since $\Lang{\varphi} = \Lang{\mathcal{A}_\varphi}$ holds by assumption, $\computation{t}{\gamma}\not\models\varphi$ follows. Hence, (i) holds and therefore the claim follows.
	
	Second, let (i) or (ii) hold. If (ii) holds, then $\ddominatesSequence{s}{t}{\gamma}{\mathcal{A}_{\varphi}}$ for $\mathcal{A}_\varphi$ follows immediately. Thus, let~(i) hold, \ie, we have $\computation{t}{\gamma}\not\models\varphi$. Then, since $\Lang{\varphi} = \Lang{\mathcal{A}_\varphi}$ holds by assumption,~$\mathcal{A}_\varphi$ rejects $\computation{t}{\gamma}$ and hence for all run trees of $\mathcal{A}_\varphi$ induced by $\computation{t}{\gamma}$, there is a branch that visits infinitely many rejecting states. 
	Let $r$ be some run tree of $\mathcal{A}_\varphi$ induced by $\computation{t}{\gamma}$. By \Cref{lem:runtree_induces_strategy_alt}, there exists a strategy $\mu$ for \alt in the \iddominance game that reflects the existential choices in $\mathcal{A}_\varphi$ for $\computation{t}{\gamma}$ defined by $\mu$. Moreover, we have $\mathcal{R}(r) = \{ \hat{\rho}^a \mid \rho \in \mathcal{P}_{\mu} \}$, where $\mathcal{R}(r)$ is the set of infinite branches of $r$ and where $\mathcal{P}_{\mu}$ is the set of initial plays that are consistent with~$\mu$.
	Note that only the part of $\mu$ controlling the existential choices of $\mathcal{A}_\varphi$ for $\computation{t}{\gamma}$ is relevant for this property. Thus, in fact, there are strategies $\mu^r$ for all run trees $r$ of $\mathcal{A}_\varphi$ induced by $\computation{t}{\gamma}$ that coincide for the other part of a strategy for \alt, \ie, the universal choices of $\mathcal{A}_\varphi$ for $\computation{s}{\gamma}$.
	Let $\mathcal{M}$ be the set of such strategies $\mu^r$ of all run trees $r$.
	As shown above, the sequences of alternative states in consistent plays of such strategies $\mu^r \in \mathcal{M}$ coincide with branches of $r$.
	Thus, since every $r$ contains a branch $b$ that visits infinitely many rejecting states, there also exists an initial play $\rho^r$ of the \iddominance game that is consistent with $\mu^r$ and which contains infinitely many rejecting alternative states. Note that since the number of rejecting alternative states is only affected by the alternative states of the play and since all $\mu^r$ coincide on the universal choices of $\mathcal{A}_\varphi$ for $\computation{s}{\gamma}$, there are, in particular, such plays $\rho^r$ that all coincide on the dominant states. 
	Moreover, there is a set $\mathcal{P}$ of such plays such that for every two plays $\rho, \rho' \in \mathcal{P}$ that coincide in the alternative states up to index $j$ as well as in the previous decisions for the alternative states in the current round $j+1$ of the game,~$\rho$ and~$\rho'$ coincide on the universal decision for the alternative state in round $j+1$ as well: suppose that this is not the case. Then, there is a finite prefix $w \in V^\omega$ of a play that coincides with~$\rho$ and~$\rho'$ up to point in time $|w|$ and that requires a universal choice between options~$u$ and~$u'$ in the next step. Moreover, suppose that $u$ is the correct extension of~$w$ for a play $\rho$, while~$u'$ is the correct one for a play $\rho'$, \ie, the respective other choice does not yield a play with infinitely many rejecting alternative states. But then, there is also the run tree~$r$ that, depending on the universal choice $u$ vs.\ $u'$ makes the existential choices that causes a play with only finitely many rejecting alternative states, \ie, $\rho$ for $u'$ and $\rho'$ for $u$. Since this is the case for all such situations and since there are run trees for all possible combinations of existential choices, there thus exists a run tree whose branches all visit only finitely many rejecting states; contradicting the assumption.
	Hence, there indeed exists such a set $\mathcal{P}$ of plays of the \iddominance game that (i) all contain infinitely many rejecting alternative states,~(ii)~coincide on the dominant states, and (iii) where for every two plays $\rho, \rho' \in \mathcal{P}$ that coincide in the alternative states up to index $j$ as well as in the previous decisions for the alternative states in the current round $j+1$ of the game, $\rho$ and $\rho'$ coincide on the universal decision for the alternative state in round $j+1$ as well.
		Thus, in particular, for every finite prefix of a play in $\mathcal{P}$, the next universal decision of $\mathcal{A}_\varphi$ for $\computation{t}{\gamma}$ can always be made solely based on the information about the history.
	Hence, we construct a \emph{strategy}~$\tau'$ for the universal choices of $\mathcal{A}_\varphi$ for $\computation{t}{\gamma}$ from $\mathcal{P}$ by defining the respective choice defined by the plays in $\mathcal{P}$ for every finite prefix.
	But then, since all plays in $\mathcal{P}$ contain infinitely many rejecting alternative states, every initial play that is consistent with $\tau'$ does so as well.
	Since the existential choices of $\mathcal{A}_\varphi$ for $\computation{s}{\gamma}$ do not influence the alternative states of a play, it follows that for all strategies~$\tau$ for \dom that coincides with~$\tau'$ on the universal choices for $t$, all consistent initial plays contain infinitely many rejecting alternative states. Thus, all such strategies $\tau$ are winning strategies for \dom and thus $\ddominatesSequence{s}{t}{\gamma}{\mathcal{A}_{\varphi}}$ holds for~$\mathcal{A}_\varphi$.
\end{proof}

\subsection{Soundness and Completeness of \texorpdfstring{$\ddACA{\varphi}$}{B{A}} (Proof of \texorpdfstring{\Cref{lem:soundness_completeness_sequence}}{Lemma 10})}

Let $\varphi$ be an LTL formula. Let $\mathcal{A}_{\varphi}$ and $\mathcal{A}_{\neg\varphi}$ be two alternating co-Büchi automata with $\Lang{\varphi} = \Lang{\mathcal{A}_\varphi}$ and $\Lang{\neg\varphi} = \Lang{\mathcal{A}_{\neg\varphi}}$. Let $\ddACA{\varphi}$ be the ACA constructed from~$\mathcal{A}_{\varphi}$ and $\mathcal{A}_{\neg\varphi}$ as described in \Cref{def:aca_dd}.
We prove that $\ddACA{\varphi}$ is sound and complete in the sense that it recognizes whether or not a strategy $s$ for a process $p_i$ \iddominates a strategy $t$ for $p_i$ on an input sequence $\gamma \in (2^\inputs{i})^\omega$.
First, we prove soundness of $\ddACA{\varphi}$, \ie, that if $\ddACA{\varphi}$ accepts a sequence $\computation{s}{\gamma} \cup \primeSequence{\computation{t}{\gamma} \cap \outputs{i}}$, then $\ddominatesSequence{s}{t}{\gamma}{\mathcal{A}_{\varphi}}$ holds:

\begin{lemma}\label{lem:soundness_sequence}
	Let $\varphi$ be an LTL formula.
	Let $\mathcal{A}_{\varphi}$ and $\mathcal{A}_{\neg\varphi}$ be ACAs with $\Lang{\varphi} = \Lang{\mathcal{A}_\varphi}$ and $\Lang{\neg\varphi} = \Lang{\mathcal{A}_{\neg\varphi}}$. Let $\ddACA{\varphi}$ be the ACA constructed from~$\mathcal{A}_{\varphi}$ and $\mathcal{A}_{\neg\varphi}$ according to~\Cref{def:aca_dd}. Let $s$ and $t$ be strategies for process $p_i$ and let $\gamma \in (2^\inputs{i})^\omega\!$. Let $\sigma \in (2^{\variables{i} \cup \primedOutputs{i}})^\omega\!$ with $\sigma := \computation{s}{\gamma} \cup \primeSequence{\computation{t}{\gamma}\cap\outputs{i}}$. If $\ddACA{\varphi}$ accepts $\sigma$, then $\ddominatesSequence{s}{t}{\gamma}{\mathcal{A}_{\varphi}}$ holds.
\end{lemma}
\begin{proof}
	Suppose that $\ddACA{\varphi}$ accepts $\sigma$. Then, there exists a run tree $r$ of $\ddACA{\varphi}$ induced by~$\sigma$ whose branches all visit only finitely many rejecting states.
	By definition, $r$ defines the existential choices in~$\ddACA{\varphi}$ for~$\sigma$. 
	Thus, in particular, $r$ defines the choice in the initial state $(q_0,q_0,\top)$ for, intuitively, either \enquote{entering} the ACA $\mathcal{A}_{\neg\varphi}$ for the negated formula or for \enquote{entering} the product automaton part of $\ddACA{\varphi}$.
	
	First, suppose that $r$ defines to \enquote{enter} the ACA $\mathcal{A}_{\neg\varphi}$. Then, by construction of $\ddACA{\varphi}$, there is a run tree $\tilde{r}$ of $\mathcal{A}_{\neg\varphi}$ that only differs from $r$ in the labeling of the root: in $\tilde{r}$, the root is labeled with $q_0$, while it is labeled with $(q_0,q_0,\top)$ in $r$. Thus, by definition of the rejecting states $\ddACAComponent{F}$ of $\ddACA{\varphi}$, $\tilde{r}$ visits only finitely many rejecting states as well.
	Moreover, since $\mathcal{A}_{\neg\varphi}$ is an ACA with alphabet $2^\variables{i}$, the successors in $r$ only depend on the valuations of the variables in $\sigma \cap \primedVariables{i}$. 
	By assumption, all branches of $r$ visit only finitely many rejecting states. Thus, all branches of the corresponding run tree $\tilde{r}$ of $\mathcal{A}_{\neg\varphi}$ visit only finitely many rejecting states as well. 
	Hence, $\mathcal{A}_{\neg\varphi}$ accepts $\unprime{\sigma \cap \primedVariables{i}}$ and thus $\unprime{\sigma \cap \primedVariables{i}} \in \Lang{\mathcal{A}_{\neg\varphi}}$. Since $\Lang{\neg\varphi} = \Lang{\mathcal{A}_{\neg\varphi}}$ by assumption, $\unprime{\sigma\cap\primedVariables{i}}\not\models\varphi$ follows.
	By definition, we have $\unprime{\sigma\cap\primedVariables{i}} = \computation{t}{\gamma}$. Therefore, $\computation{t}{\gamma}\not\models\varphi$ holds.
	With \Cref{lem:disjunctive_ddominance} we thus obtain immediately that $\ddominatesSequence{s}{t}{\gamma}{\mathcal{A}_{\varphi}}$ holds for~$\mathcal{A}_{\varphi}$; proving the claim.
	
	Second, suppose that $r$ defines to \enquote{enter} the product automaton part of $\ddACA{\varphi}$. Then, we construct a strategy $\tau$ for \dom in the \iddominance game from $r$ as follows.
	Let $\eta \cdot \delta$ be a finite sequence of positions with $\eta \in V^*$ and $\delta \in V$. We only define~$\tau$ explicitly on sequences $\eta \cdot \delta$ that can occur in the \iddominance game $(\aca{\varphi},\computation{t}{\gamma},\computation{s}{\gamma})$ and where $\delta$ is controlled by \dom; on all other sequences we define $\tau(\eta\cdot\delta) = v$ for some arbitrary $v \in V$ that is a valid extension of $\eta \cdot \delta$. Thus, in the following we assume that~$\eta\cdot\delta$ is a prefix that can occur in the game and that $\delta$ is of the form $((p,q,c),j)$ or $((p,q,c,q'),j)$. We map~$\eta \cdot \delta$ to a prefix of a branch of $r$ if there is a compatible one: a compatible branch $b$ of $r$ agrees with the finite projected play $\hat{\eta}$ up to point in time~$|\eta|$. Note her that, slightly misusing notation, we apply the definition of a projected play also to the finite prefix $\eta$ of a play.
	Moreover, no matter whether $\delta$ is of the form $((p,q,c),j)$ or $((p,q,c,q'),j)$, we have  $b_{|\eta|+1} = (p,q,m)$ for some $m \in \{\top,\bot\}$.
	If there is no compatible branch in $r$, we again define $\tau(\eta\cdot\delta) = v$ for some arbitrary $v \in V$ that is a valid extension of $\eta \cdot \delta$.
	Otherwise, the successors of $(p,q,m)$ in $b$ define the choice of $\tau$: by definition, the set $\mathcal{S}$ of successors of $(p,q,m)$ satisfies $\ddACAComponent{\delta}((p,q,m),\sigma_{|\eta|+1})$. Thus, for all $c \in \delta(p,\unprime{\sigma_{|\eta|+1} \cap \primedVariables{i}})$, there is some $c' \in \delta(q,\sigma_{|\eta|+1} \cap \variables{i})$ such that for all $q' \in c'$, there is some $p' \in c$ such that $\successor{p'}{q'}{m} \in \mathcal{S}$ holds. Note here that we do not distinguish between $(q_0,q_0,\top)$ and other states $(p,q,m)$ since, by assumption, $r$ defines the choice of entering the product automaton part of $\ddACA{\varphi}$ and thus the choice of the second disjunct for $(q_0,q_0,\top)$ which coincides with $\ddACAComponent{\delta}$ for other states $(p,q,m)$.
	If $\delta = ((p,q,c),j)$, we thus define $\tau(\eta\cdot\delta) = ((p,q,c_c'),j)$, where the choice of~$c'$ is based on $c$.
	If $\delta = ((p,q,c,q'),j)$, then we define $\tau(\eta\cdot\delta) = ((p',q'),j+1)$, where the choice of $p'$ is based on $c$, $c'$, and $q'$.
		
	It remains to show that $\tau$ is winning from the initial position $v_0$. Let $\rho$ be some initial play that is consistent with $\tau$. 
	Then, by construction of $\tau$, there is a branch $b$ of $r$ that coincides with the projected play $\hat{\rho}$ on $p$ and $q$, \ie, we have  $\hat{\rho} = \hat{b}$, where $\hat{b}$ is the sequence obtained from $b$ when removing the marking $m$ from all nodes $(p,q,m)$.
	By assumption, all branches of $r$ visit only finitely many rejecting states. Thus, in particular the branch~$b$ with $\hat{b} = \hat{\rho}$ visits only finitely many rejecting states.
	Hence, by construction of $\ddACA{\varphi}$ and since, by assumption, we only consider the product automaton part of $\ddACA{\varphi}$, $b$ thus visits only finitely many states of the form $(p,q,\bot)$.
	Moreover, by definition of $\vartheta$, we only have $m = \bot$ for a state $(p,q,m)$ at position $j \in \mathbb{N}$ of $b$ if either (i) $p \not \in F$ and $q \in F$ holds, or if~(ii)~$p'\not\in F$ and $q' \in F$ holds for $(p',q',m')$ at some position $j' < j$ of $b$ and $p'' \not \in F$ holds for $(p'',q'',m'')$ at all positions $j''$ with $j' \leq j'' \leq j$. Therefore, since $b$ visits only finitely many states of the form $(p,q,\bot)$, there are only finitely many points in time, where $b$ visits a rejecting dominant state while it does not visit a rejecting alternative state, and for all these points in time there are only finitely many following steps until a rejecting alternative state is visited.
	Thus, in particular, $\projection{\hat{b}_j}{2} \in F \rightarrow \exists j' \geq j.~ \projection{\hat{b}_{j'}}{1}\in F$ holds for all points in time $j \in \mathbb{N}$.
	Since $\hat{b} = \hat{\rho}$, it thus follows that $\rho \in \win$ holds; proving the claim.
\end{proof}

Next, we prove completeness of $\ddACA{\varphi}$, \ie, that if $s$ \iddominates $t$ on input $\gamma$, then~$\ddACA{\varphi}$ accepts a sequence $\computation{s}{\gamma} \cup \primeSequence{\computation{t}{\gamma} \cap \outputs{i}}$.

\begin{lemma}\label{lem:completeness_sequence}
	Let $\varphi$ be an LTL formula.
	Let $\mathcal{A}_{\varphi}$ and $\mathcal{A}_{\neg\varphi}$ be ACAs with $\Lang{\varphi} = \Lang{\mathcal{A}_\varphi}$ and $\Lang{\neg\varphi} = \Lang{\mathcal{A}_{\neg\varphi}}$. Let $\ddACA{\varphi}$ be the ACA constructed from~$\mathcal{A}_{\varphi}$ and $\mathcal{A}_{\neg\varphi}$ according to~\Cref{def:aca_dd}. Let $s$ and $t$ be strategies for process $p_i$ and let $\gamma \in (2^\inputs{i})^\omega\!$. Let $\sigma \in (2^{\variables{i} \cup \primedOutputs{i}})^\omega\!$ with $\sigma := \computation{s}{\gamma} \cup \primeSequence{\computation{t}{\gamma}\cap\outputs{i}}$. If $\ddominatesSequence{s}{t}{\gamma}{\mathcal{A}_{\varphi}}$ holds, then $\ddACA{\varphi}$ accepts $\sigma$.
\end{lemma}
\begin{proof}
	Suppose that $\ddominatesSequence{s}{t}{\gamma}{\mathcal{A}_{\varphi}}$ holds for~$\mathcal{A}_{\varphi}$. Then, by \Cref{lem:disjunctive_ddominance}, either (i) $\computation{t}{\gamma}\not\models\varphi$ holds or (ii) $\ddominatesSequence{s}{t}{\gamma}{\mathcal{A}_{\varphi}}$ and for the winning strategy $\tau$ of \dom in the \iddominance game we have for every initial play $\rho$ that is consistent with $\tau$ that there is a point in time $k$ such that $\projq{\rho_{k'}} \not\in F$ for all $k' \geq k$. We distinguish two cases.
	
	First, suppose that (i) holds. Then, we have $\computation{t}{\gamma}\in\Lang{\neg\varphi}$ and thus, since $\Lang{\neg\varphi} = \Lang{\mathcal{A}_{\neg\varphi}}$ by assumption, $\computation{t}{\gamma}\in\Lang{\mathcal{A}_{\neg\varphi}}$ holds. Thus, there exists a run tree $r$ of $\mathcal{A}_{\neg\varphi}$ induced by $\computation{t}{\gamma}$ whose branches all visit only finitely many rejecting states.
	By construction of $\ddACA{\varphi}$, there exists a corresponding run tree $\tilde{r}$ of $\ddACA{\varphi}$ that only differs from $r$ in the labeling of the root: in $r$, the root is labeled with $q_0$, while it is labeled with $(q_0,q_0,\top)$ in~$\tilde{r}$. Hence, by definition of the rejecting states $\ddACAComponent{F}$ of $\ddACA{\varphi}$, $\tilde{r}$ visits only finitely many rejecting states as well. Moreover, since $\mathcal{A}_{\neg\varphi}$ is an ACA with alphabet $2^\variables{i}$ and by construction of the transition function $\ddACAComponent{\delta}$ of $\ddACA{\varphi}$, the successors in $\tilde{r}$ only depend on the valuations of the variables in $\primedVariables{i}$. Thus, all sequences $\sigma' \in (2^{\variables{i} \cup \primedOutputs{i}})^\omega$ with $\sigma' \cap \primedVariables{i} = \primeSequence{\computation{t}{\gamma}}$ induce the run tree $\tilde{r}$. Therefore, in particular $\sigma$ does. Thus, $\sigma$ induces a run tree on $\ddACA{\varphi}$, namely $\tilde{r}$, that visits only finitely many rejecting states and therefore $\ddACA{\varphi}$ accepts $\sigma$.
	
	Second, suppose that (ii) holds. Then, there exists a winning strategy $\tau$ for \dom in the \iddominance game $(\mathcal{A}_\varphi,\computation{t}{\gamma},\computation{s}{\gamma})$. Let $\mathcal{P}_\tau$ be the set of initial plays that are consistent with $\tau$.
	Let $f: (Q \times Q)^\omega \rightarrow (Q \times Q \times \{\top,\bot\})^\omega$ be a function that, given an infinite sequence~$\chi$ of tuples $(p,q)$, returns an extended sequence $\chi'$ that is incrementally defined as follows: for the initial point in time, let $\chi'_0 := (p,q,\top)$ if $\chi_0 = (p,q)$. For a point in time $j > 0$, let $\chi'_j := \successor{p'}{q'}{m}$ if $\chi'_{j-1} = (p,q,m)$ and $\chi_j = (p',q')$. Here,~$\vartheta$ denotes the corresponding function used in \Cref{def:aca_dd}.
	We construct a $Q$-labeled tree~$(\mathcal{T},\ell)$ from $\tau$ as follows by defining the labeling of the root as well as of the successors of all nodes.
	The labeling of the root $\varepsilon$ of $\mathcal{T}$ is defined by $\ell(\varepsilon)=(q_0,q_0,\top)$. For a node $x \in \mathcal{T}$ with depth $j = |x|$, we define the labeling of the successor nodes of $x$ such that $\{ \ell(x') \mid x' \in \children{x} \} = \{ f(\hat{\rho}_{j+1}) \mid \rho \in \mathcal{P}_\tau \land \compatible{\rho}{x} \}$, where $\compatible{\rho}{x}$ denotes that~$\rho$ and~$x$ are compatible in the sense that for all $0 \leq j' \leq j$, we have $f(\hat{\rho}_{j'}) = \ell(a_{j'})$, where~$a$ is the unique finite sequence of nodes that, starting from $\varepsilon$, reaches $x$. 
	Next, we show that $(\mathcal{T},\ell)$ is a run tree of $\ddACA{\varphi}$ induced by $\sigma$.
	For the sake of readability, let $\sigma^s := \computation{s}{\gamma}$ and $\sigma^t := \computation{t}{\gamma}$.
	Let $\mathcal{S}_a := \{ \hat{\rho}^a_{j+1} \mid \rho \in \mathcal{P}_\tau \land \compatible{\rho}{x} \}$ for some node $x\in\mathcal{T}$. Then, by construction of the \iddominance game, we know that~$\mathcal{S}_a$ satisfies $\bigwedge_{c \in \delta(p,\primeSequence{\sigma^t_j)}} \bigvee_{p' \in c} p'$, where $p := \projection{\ell(x)}{1}$. Thus, intuitively, the alternative states of an initial play that is consistent with~$\tau$ evolve according to a run of $\mathcal{A}'$ induced by $\computation{t}{\gamma}$, where $\mathcal{A}'$ is the ACA obtained from~$\aca{\varphi}$ by dualizing the transition function, \ie, by swapping conjunctions and disjunctions.
	Let $\mathcal{S}_d := \{ \hat{\rho}^d_{j+1} \mid \rho \in \mathcal{P}_\tau \land \compatible{\rho}{x} \}$ for some node $x\in\mathcal{T}$. 
	Then, by construction of the \iddominance game, we know that~$\mathcal{S}_d$ satisfies $\bigvee_{c' \in \delta(q,\sigma^s_j)} \bigwedge_{q' \in c'} q'$, where $q := \projection{\ell(x)}{2}$. Hence, intuitively, the dominant states of of an initial play that is consistent with~$\tau$ evolve according to a run of~$\aca{\varphi}$ induced by $\computation{s}{\gamma}$.
	From these observations, it follows that $\{ \projection{\rho_{5j}}{1} \mid \rho \in \mathcal{P}_\tau \land \compatible{\rho}{x} \}$ satisfies $\bigwedge_{c \in \delta(p,\primeSequence{\sigma^t_j)}} \bigvee_{c' \in \delta(q,\sigma^s_j)} \bigwedge_{q' \in c'} \bigvee_{p' \in c} (p',q')$, where $p := \projection{\ell(x)}{1}$ and $q := \projection{\ell(x)}{2}$.
	Therefore, by construction of $\ddACA{\varphi}$, we know that, for every node $x \in \mathcal{T}$, $\{ f(\hat{\rho}_{j+1}) \mid \rho \in \mathcal{P}_\tau \land \compatible{\rho}{x} \}$ satisfies $\ddACAComponent{\delta}((p,q,m),\sigma_j)$, where $(p,q,m) := \ell(x)$. Hence, $(\mathcal{T},\ell)$ is indeed a run tree of $\ddACA{\varphi}$ induced by $\sigma$.
	Since $\tau$ is a winning strategy for \dom, we have $\rho \in \win$ for all initial plays $\rho$ that are consistent with~$\tau$. Hence, for all such plays $\rho \in \mathcal{P}_\tau$ and all points in time $j \in \mathbb{N}$, it holds that if $\projq{\rho_j} \in F$ holds, then we have $\projp{\rho_{j'}} \in F$ for some point in time $j' \geq j$ as well.
	Moreover, by assumption, for every initial play $\rho$ that is consistent with $\tau$, we have that there is a point in time $k$ such that $\projq{\rho_{k'}} \not\in F$ for all $k' \geq k$.
	Thus, there are only finitely many points in time at which~$\rho$ visits a rejecting dominant state and for all these points in time it holds that a rejecting alternative state occurs in $\rho$ at the very same point in time or later.
	Therefore, by construction of $(\mathcal{T},\ell)$ and $f$, we obtain that there are only finitely many nodes $x \in \mathcal{T}$ with $\ell(x) = (p,q,\bot)$ for some $p,q\in Q$.
	Therefore, since only states of the form $(p,q,m)$ are reached and since for these states the ones with mark $\bot$ are the only rejecting ones of $\ddACA{\varphi}$, all branches of $(\mathcal{T},\ell)$ visit only finitely many rejecting states. Hence, since $(\mathcal{T},\ell)$ is a run tree of $\ddACA{\varphi}$ induced by~$\sigma$, $\ddACA{\varphi}$ accepts~$\sigma$.
\end{proof}

From \Cref{lem:soundness_sequence,lem:completeness_sequence}, the claim of \Cref{lem:soundness_completeness_sequence} then follows immediately.

\subsection{Miyano-Hayashi for co-Büchi Automata (Proof of \texorpdfstring{\Cref{thm:miyano-hayashi_universal}}{Lemma 11})}

The Miyano-Hayashi algorithm~\cite{MiyanoH84} is a well-known technique for translating alternating Büchi automata into equivalent nondeterministic Büchi automata. It introduces an exponential blowup: the resulting NBA is of exponential size in the number of states of the initial ABA.
In this paper, we consider co-Büchi automata instead of Büchi automata. Hence, we need a translation from alternating co-Büchi automata to universal co-Büchi automata.
Recall from \Cref{app:preliminaries} that the Büchi and co-Büchi acceptance conditions as well as nondeterministic and universal branching are dual. Thus, we can reuse the Miyano-Hayashi algorithm for co-Büchi automata by making use of the duality; proving \Cref{thm:miyano-hayashi_universal}:

\begin{proof}
	Let $\mathcal{A} = (Q, Q_0, \delta, F)$.
	Let $\mathcal{A}_d = (Q^d, Q^d_0,\delta^d,F^d)$ be the dual automaton of $\mathcal{A}$, \ie, the ABA with $Q^d = Q$, $Q^d_0 = Q_0$, $F^d = F$, and $\delta^d(u,\boldsymbol{i}) = \bigwedge_{c \in \delta(u,\boldsymbol{i})} \bigvee_{u' \in c} u'$. Then $\mathcal{L}(\mathcal{A}_d) = \overline{\mathcal{L}(\mathcal{A})}$ holds due to the duality of nondeterministic and universal branching as well as of the Büchi and co-Büchi acceptance condition. 
	As shown by Miyano and Hayashi~\cite{MiyanoH84}, there exists a nondeterministic Büchi automaton $\mathcal{B}'$ with $\mathcal{O}(2^{|Q^d|})$ states and with $\mathcal{L}(\mathcal{B}') = \mathcal{L}(\mathcal{A}^d)$. Let~$\mathcal{B}$ be the dual automaton of $\mathcal{B}'$, \ie, the universal co-Büchi automaton that is a copy of~$\mathcal{B}'$, but where the nondeterministic transitions are interpreted as universal ones and where the accepting states are interpreted as rejecting states. Then, $\mathcal{B}$ has $\mathcal{O}(2^{|Q^d|})$ states and we have $\mathcal{L}(\mathcal{B}) = \overline{\mathcal{L}(\mathcal{B}')}$. Since $\mathcal{L}(\mathcal{B}') = \mathcal{L}(\mathcal{A}^d) = \overline{\mathcal{L}(\mathcal{A})}$ holds, we obtain $\mathcal{L}(\mathcal{B}) = \mathcal{L}(\mathcal{A})$. Thus,~$\mathcal{B}$ is the desired universal co-Büchi automaton.
\end{proof}

\subsection{Correctness of Universal Projection (Proof of \texorpdfstring{\Cref{lem:universal_projection}}{Lemma 13})}

Universal projection allows for abstracting away variables from a universal automaton. The resulting automaton then accepts a sequence $\sigma$ if, and only if, the initial automaton accepts all sequences $\sigma'$ that extend $\sigma$, \ie, with $\sigma' \cap X = \sigma$, where $X$ is the set we projected to. Note that this only holds for universal automata since it relies on the universal branching of the automaton. We prove that the above result (and thus \Cref{lem:completeness_sequence}) holds:

\begin{proof}
	By construction of $\pi_X$, $\sigma$ induces a path $\pi$ in $\pi_X(\mathcal{A})$ if, and only if, there is a sequence $\sigma' \in (2^\Sigma)^\omega$ with $\sigma' \cap X = \sigma$ that induces the same path $\pi$ in $\mathcal{A}$.
	
	First, suppose that $\pi_X(\mathcal{A})$ accepts $\sigma$. Then, by definition of universal co-Büchi automata, all paths $\pi$ of $\pi_X(\mathcal{A})$ induced by $\sigma$ visit rejecting states only finitely often. Suppose that there is a $\sigma' \in (2^\Sigma)^\omega$ with $\sigma' \cap X = \sigma$ that is rejected by $\mathcal{A}$. Hence, there is a path $\pi'$ of $\mathcal{A}$ induced by $\sigma'$ that visits infinitely many rejecting states. But then, as shown above, $\pi'$ is a path of $\pi_X(\mathcal{A})$ induced by $\sigma$ as well; contradicting the assumption that all paths of $\pi_X(\mathcal{A})$ induced by $\sigma$ visit rejecting states only finitely often.
	
	Second, suppose that $\mathcal{A}$ accepts all $\sigma' \in (2^\Sigma)^\omega$ with $\sigma' \cap X = \sigma$. Then, by definition of universal co-Büchi automata, all paths~$\pi$ of $\mathcal{A}$ induced by some $\sigma' \in (2^\Sigma)^\omega$ with $\sigma' \cap X = \sigma$ visit rejecting states only finitely often. Suppose that $\pi_X(\mathcal{A})$ rejects~$\sigma$. Then, there is a path~$\pi$ in $\pi_X(\mathcal{A})$ induced by $\sigma$ that visits rejecting states infinitely often. But then, as shown above, there is some $\sigma' \in (2^\Sigma)^\omega$ with $\sigma' \cap X = \sigma$ such that $\pi$ is also a path of $\mathcal{A}$ induced by $\sigma'$; contradicting the assumption that all paths $\pi$ of $\mathcal{A}$ induced by some $\sigma' \in (2^\Sigma)^\omega$ with $\sigma' \cap X = \sigma$ visit rejecting states only finitely often.
\end{proof}

\subsection{Soundness and Completeness of \texorpdfstring{$\ddUCA{\varphi}$}{A{dd}} (Proof of \texorpdfstring{\Cref{thm:soundness_completeness_ddUCA}}{Theorem 15})}

From the construction of the universal co-Büchi automaton $\ddUCA{\varphi}$ as well as of \Cref{thm:soundness_completeness_ddUCA}, \ie, the result that $\ddACA{\varphi}$ determines whether or not a strategy \iddominates another strategy on an input sequence, it now follows that $\ddUCA{\varphi}$ is sound and complete in the sense that it recognizes \iddominant strategies:

\begin{proof}
	Let $\ddACA{\varphi}$ and $\ddUCAnonProj{\varphi}$ be the intermediate automata from which $\ddUCA{\varphi}$ is constructed. Note that $\ddUCA{\varphi}$ is a UCA over alphabet $2^\variables{i}$, while $\ddUCAnonProj{\varphi}$ and $\ddACA{\varphi}$ are alternating and universal co-Büchi automata, respectively, over alphabet $2^{\variables{i} \cup \primedOutputs{i}}$.
	Since $\ddUCA{\varphi}$ is the universal projection of $\ddUCAnonProj{\varphi}$ to $\variables{i}$, we obtain with \Cref{lem:universal_projection} that $\ddUCA{\varphi}$ accepts a sequence $\sigma \in (2^\variables{i})^\omega$ if, and only if, $\ddUCAnonProj{\varphi}$ accepts all sequences $\sigma' \in (2^{\variables{i} \cup \primedOutputs{i}})^\omega$ with $\sigma' \cap \variables{i} = \sigma$. By \Cref{thm:miyano-hayashi_universal}, we have $\mathcal{L}(\ddUCAnonProj{\varphi}) = \mathcal{L}(\ddACA{\varphi})$. Thus, $\ddUCA{\varphi}$ accepts a sequence $\sigma \in (2^\variables{i})^\omega$ if, and only if, $\ddACA{\varphi}$ accepts all sequences $\sigma' \in (2^{\variables{i} \cup \primedOutputs{i}})^\omega$ with $\sigma' \cap \variables{i} = \sigma$.
	Let $\mathcal{S} := \{ \computation{s}{\gamma} \mid \gamma \in (2^\inputs{i})^\omega \}$.
	Then $\ddUCA{\varphi}$ accepts all sequences $\sigma \in \mathcal{S}$ if, and only if $\ddACA{\varphi}$ accepts all sequences $\sigma' \cap \variables{i} \in \mathcal{S}$.
	
	Let $\mathcal{E}$ be the set of all such extended sequences $\sigma'$, \ie, $\mathcal{E} := \{ \sigma' \mid \sigma' \cap \variables{i} \in \mathcal{S} \}$. Intuitively, we have $\sigma' \in \mathcal{E}$ if, and only if, $\sigma'$ extends some computation of $s$ with primed output variables. 
	Let $\mathcal{M} = \{ \computation{s}{\gamma} \cup \primeSequence{\computation{t}{\gamma}\cap\outputs{i}} \mid t \text{ is a strategy for } p_i \text{ and } \gamma \in (2^\inputs{i})^\omega \}$. By construction, all sequences $\sigma'' \in \mathcal{M}$ are extensions of computations of $s$ with primed output variables, \ie, $\mathcal{M} \subseteq \mathcal{E}$. Moreover, all sequences of valuations of output variables of $p_i$ can be produced by some strategy for $p_i$. Hence, by priming them, we obtain all sequences of valuations of primed output variables of $p_i$. Thus, $\mathcal{E} \subseteq \mathcal{M}$ holds as well.
	Therefore, it follows that $\ddUCA{\varphi}$ accepts all sequences $\sigma \in \mathcal{S}$ if, and only if $\ddACA{\varphi}$ accepts all sequences $\sigma' \in \mathcal{M}$.
	By construction of $\mathcal{S}$ and $\mathcal{M}$, it thus follows immediately with \Cref{lem:soundness_completeness_sequence} that $\ddUCA{\varphi}$ accepts $\computation{s}{\gamma}$ for all $\gamma \in (2^\inputs{i})^\omega$, if, and only if $s$ is \iddominant for~$\mathcal{A}_\varphi$.
\end{proof}

\subsection{Size of \texorpdfstring{$\ddUCA{\varphi}$}{A{dd}} (Proof of \texorpdfstring{\Cref{thm:automaton_size}}{Lemma 16})}

Putting together the previous results, we can now prove \Cref{thm:automaton_size}, showing that the universal co-Büchi automaton $\ddUCA{\varphi}$ is of size exponential in the squared length of the LTL formula $\varphi$:

\begin{proof}
	Given an LTL formula $\varphi$, there are, by \cite{MullerSS88,Vardi94}, ACAs $\mathcal{A}_\varphi = (Q,Q_0,\delta,F)$ and $\mathcal{A}_{\neg\varphi} = (Q^c,Q^c_0,\delta^c,F^c)$, both of size $\mathcal{O}(|\varphi|)$, with $\Lang{\mathcal{A}_\varphi} = \Lang{\varphi}$ and $\Lang{\mathcal{A}_{\neg\varphi}} = \Lang{\neg\varphi}$.
	By \Cref{thm:soundness_completeness_ddUCA}, the automaton $\ddUCA{\varphi}$ constructed according to \Cref{def:UCA_construction_delayed_dominance} satisfies the property that~$\ddUCA{\varphi}$ accepts $\computation{s}{\gamma}$ for some strategy $s$ and for all $\gamma \in (2^\inputs{i})^\omega$ if, and only if, $s$ is \iddominant for $\mathcal{A}_\varphi$.
	Let $\ddACA{\varphi}$ and $\ddUCAnonProj{\varphi}$ be the intermediate automata from which $\ddUCA{\varphi}$ is constructed. By construction, $\ddACA{\varphi}$ is of size $\mathcal{O}(|Q|^2+|Q^c|)$. By construction of $\ddUCAnonProj{\varphi}$ and by \Cref{thm:miyano-hayashi_universal}, $\ddUCAnonProj{\varphi}$ is of size $\mathcal{O}(2^m)$, where $m$ is the number of states of $\ddACA{\varphi}$. Hence,~$\ddUCAnonProj{\varphi}$ has $\mathcal{O}(2^{|Q|^2+|Q^c|})$ states. Since the universal projection does not affect the size of an automaton as it only alters the transition relation,~$\ddUCA{\varphi}$ has $\mathcal{O}(2^{|Q|^2+|Q^c|})$ states as well. Since both $\mathcal{A}_\varphi$ and $\mathcal{A}_{\neg\varphi}$ have $\mathcal{O}(|\varphi|)$ states the claim follows.
\end{proof}

\end{document}